\documentclass[aps,pra,twocolumn,superscriptaddress,10pt]{revtex4-1}
\usepackage{braket}
\usepackage{bbm}
\usepackage{amsmath}
\usepackage{amssymb}
\usepackage{amsthm}
\usepackage{graphicx} 
\usepackage{xcolor}
\usepackage{mathtools}
\usepackage{dutchcal}
\usepackage[colorlinks=true,citecolor=blue,linkcolor=blue,urlcolor=blue]{hyperref}

\theoremstyle{plain}
\newtheorem{theorem}{Theorem}

\newtheorem{lemma}[theorem]{Lemma}
\newtheorem{corollary}[theorem]{Corollary}

\theoremstyle{definition}
\newtheorem{definition}{Definition}

\theoremstyle{remark}

\newcommand{\aqa}{$\langle aQa ^L\rangle $ Applied Quantum Algorithms, Universiteit Leiden}
\newcommand{\lorentz}{Instituut-Lorentz, Universiteit Leiden, Niels Bohrweg 2, 2333 CA Leiden, Netherlands}
\newcommand{\liacs}{LIACS, Universiteit Leiden, Niels Bohrweg 1, 2333 CA Leiden, Netherlands}

\def\selectlanguage#1{}

\begin{document}

\title{An exact semidefinite characterization of the stoquasticity cone of permutationally invariant Bell operators}

\author{Jan Li}
\affiliation{\aqa}
\affiliation{\lorentz}

\author{Owidiusz Makuta}
\affiliation{\aqa}
\affiliation{\lorentz}

\author{Evert van Nieuwenburg}
\affiliation{\aqa}
\affiliation{\lorentz}
\affiliation{\liacs}

\author{Jordi Tura}
\affiliation{\aqa}
\affiliation{\lorentz}

\date{\today}

\begin{abstract}
Bell correlations have been detected in atomic ensembles ranging from hundreds to hundreds of thousands of particles, witnessed by the permutationally invariant (PI) Bell operators. For all the operators realized so far, the maximally violating state can be brought, in a suitable basis, into a nonnegative-amplitude form. In that basis, the Bell operator is stoquastic. The coefficients parametrising these operators that at the same time lead to stoquasticity form a convex cone, the so-called stoquasticity cone. Its characterization had remained open beyond two-body operators, as for many-body operators the number of inequalities grows with the system size $n$. In this work we show that, in the thermodynamic limit, the stoquasticity cone admits an exact semidefinite characterization, valid at any operator order. The key observation is that, along each off-diagonal, the matrix elements are the values of a single univariate polynomial. Being univariate, this condition is equivalent to a linear matrix inequality, with no relaxation hierarchy. Therefore, deciding stoquasticity becomes a single semidefinite feasibility problem. As its lowest-degree instances, this description recovers the known three-hyperplane characterization for two-body operators and yields, for three-body operators, a curved boundary.
\end{abstract}

\maketitle

\section{Introduction}

The ground state of a quantum many-body Hamiltonian encodes its essential low-energy physics, from quantum phase transitions~\cite{sachdev_quantum_2011} to the solutions of hard optimization problems suitably encoded into Hamiltonians~\cite{farhi_quantum_2001, farhi_quantum_2014, lucas_ising_2014}. A remarkable special case is the class of \emph{stoquastic} Hamiltonians, whose off-diagonal matrix elements are nonpositive in a given basis. Recall that for such operators the ground state may be chosen with nonnegative amplitudes~\cite{bravyi_complexity_2008}, so that it defines a genuine probability distribution and can, in favorable cases, be sampled by quantum Monte Carlo methods free of the notorious sign problem.
Stoquasticity does not by itself guarantee an efficient simulation~\cite{hastings_obstructions_2013}, but it removes its most severe obstruction.

A more unexpected setting in which stoquasticity has recently surfaced is that of many-body Bell nonlocality. 
The Bell operators measured in correlation experiments on ensembles ranging from hundreds to hundreds of thousands of atoms~\cite{schmied_bell_2016,engelsen_bell_2017} can, through a change of basis, be brought to stoquastic form, their ground state being the state of maximal Bell violation. Reference~\cite{li_stoquastic_2026} turned this observation into a program.
For permutationally invariant (PI) Bell operators at fixed measurement settings, the coefficients that render the operator stoquastic form a convex polyhedral cone, the so-called stoquasticity cone. 
Note that Bell operators belonging to this cone attain their optimum at a nonnegative-amplitude state, that is, at a superposition of Dicke states with nonnegative coefficients, of the kind naturally prepared in these experiments. Notably, the characterization of the stoquastic cone is nontrivial.
The cone is cut out by one linear inequality for every off-diagonal matrix element, and the number of such inequalities grows with the system size $n$. 
Only for two-body operators is the cone known to collapse to a fixed, $n$-independent set of facets, whereas for three-body and higher operators its description proliferates with no evident closed form.

In this work we show that this proliferation is only apparent, and that the stoquasticity cone of a PI Bell operator of any order admits an \emph{exact} semidefinite characterization in the thermodynamic limit.
The starting point is a structural property of the matrix elements.
Along each off-diagonal, they are the values of a single polynomial in the rescaled row index, with coefficients linear in the Bell coefficients.
The entire family of stoquasticity inequalities on that off-diagonal then amounts to requiring this one univariate polynomial to be nonnegative on an interval. By the classical Markov--Luk\'acs theorem~\cite{powers_polynomials_2000}, such a requirement is equivalent to the existence of a sum-of-squares certificate, and hence to the feasibility of a linear matrix inequality. 
Since the polynomial is univariate, no increasing-degree hierarchy is needed, unlike in the multivariate Positivstellensatz. 
The stoquasticity cone is, in other words, a projected spectrahedron, or spectrahedral shadow.

The linear matrix inequality lives in the lifted space of the Bell coefficients and the sum-of-squares Gram matrices, and the cone is its projection onto the coefficients.
Deciding whether a Bell operator is stoquastic then reduces to a single semidefinite feasibility problem, and designing stoquastic Bell operators with prescribed properties to semidefinite optimization over the cone.
Its lowest-order instances, the two- and three-body cones, require respectively no lifted variable (a polyhedral cone) and a single, provably necessary, lifted variable.

This semidefinite description both unifies and extends what was previously known. 
Its lowest-degree instances reproduce the existing results.
For two-body operators it returns the three-hyperplane characterization of Ref.~\cite{li_stoquastic_2026}, and for three-body operators a curved boundary, a quadric cone which degrades to a pair of hyperplanes in the strict thermodynamic limit.
We quantify how faithfully the thermodynamic-limit description approximates a finite ensemble, the corrections vanishing as $1/n$, and we illustrate the construction on the measurement settings used in present-day Bell experiments. 

The paper is organized as follows. In Sec.~\ref{sec:prelim} we set up stoquasticity, the PI Bell operators, and the stoquasticity cone. Section~\ref{sec:offdiagstruc} establishes the polynomial structure of the off-diagonal matrix elements, which Sec.~\ref{sec:reduction} exploits to collapse the $O(n)$ constraints to a single one on each off-diagonal. Section~\ref{sec:sdp} presents our main result, the semidefinite characterization of the resulting feasibility region together with its dual moment formulation. The two- and three-body cones are then worked out explicitly as special cases, and their geometry analyzed, in Sec.~\ref{sec:two-body} and Sec.~\ref{sec:three-body}.

\section{Stoquastic Bell operators and the stoquasticity cone}\label{sec:prelim}

We begin with stoquasticity, the property that produces ground states of a many-body Hamiltonian that admit a sign-free probabilistic representation.
We then introduce the permutationally invariant (PI) Bell operators which are the operators measured in many-body Bell experiments based on collective measurements~\cite{schmied_bell_2016,engelsen_bell_2017}. 
Imposing stoquasticity on such an operator at fixed measurement settings constrains its Bell coefficients to lie in a convex cone, the so-called stoquasticity cone~\cite{li_stoquastic_2026}, whose compact characterization is the goal of this work.

\subsection{Stoquastic operators}\label{sec:stoquastic}
 
\begin{definition}[Stoquasticity]
A Hermitian operator $\mathbf{H}$ is \emph{stoquastic with respect to
a basis} $\mathcal{B} = \{\ket{i}\}$ if its off-diagonal matrix elements are real and nonpositive,
\begin{equation}
\bra{i}\mathbf{H}\ket{j} \leq 0 \qquad \forall\; i \neq j.
\end{equation}
\end{definition}
 
The physical content of this definition lies entirely in the sign of the off-diagonal elements. 
Recall that, whenever they are nonpositive, the operator $-\mathbf{H}$ has nonnegative off-diagonal entries, so that the Perron--Frobenius theorem applied to $c\,\mathbbm{I} - \mathbf{H}$, for $c$ sufficiently large, guaranteeing a ground state whose amplitudes in $\mathcal{B}$ are nonnegative,
\begin{equation}
\braket{i|\psi_0} \geq 0 \quad \forall\; i.
\end{equation}
The ground state then carries no sign structure in $\mathcal{B}$: the imaginary-time matrix elements $\bra{i}e^{-\beta\mathbf{H}}\ket{j}\geq 0$ define, after suitable normalization, a bona fide stochastic process. This is the absence of the sign problem~\cite{troyer_computational_2005}. It does not on its own guarantee an efficient classical simulation, since sign-problem-free Monte Carlo may still fail to converge rapidly~\cite{hastings_obstructions_2013}, but it removes the single most severe obstruction to it.
Additionally, stoquastic computational problems have been constructed to investigate the boundary between classical and quantum computing.~\cite{bravyi_merlin-arthur_2006,bravyi_complexity_2008,bravyi_complexity_2010,raza_complexity_2026,kempe_complexity_2006}.

We make two remarks.
First, stoquasticity is a property relative to a \emph{fixed} basis.
Any Hermitian operator is trivially stoquastic in its own eigenbasis, where it is diagonal, so the notion acquires content only once a physically distinguished basis has been singled out, for us the Dicke basis fixed by the form of the Bell operator. 
Second, one cannot hope to rotate the sign problem away for free.
Deciding whether a local change of basis, built from $n$ single-qubit orthogonal matrices, renders a given Hamiltonian stoquastic, the problem of \emph{curing the sign problem}, is NP-complete~\cite{marvian_computational_2019}. These observations show that the stoquasticity cone is not merely a basis artifact.

This symmetry-adapted notion should not be conflated with the more familiar stoquasticity in the computational basis $\{\ket{0},\ket{1}\}^{\otimes n}$.
The latter implies the former on the symmetric sector.
For $k\neq l$ the Dicke states $\ket{D^k_n}$ and $\ket{D^l_n}$ are nonnegative superpositions of computational basis states of distinct Hamming weight, so $\bra{D^k_n}\mathbf{B}\ket{D^l_n}$ is a nonnegative combination of off-diagonal computational-basis elements and is therefore nonpositive. 
The reverse implication has no reason to hold.
The symmetric block is only an $(n+1)$-dimensional corner of the $2^n$-dimensional space, and on the remaining permutation blocks, which are built on singlet-dressed states whose minus signs break the weight-sector argument, symmetry-adapted and computational sign-positivity are unrelated.

\subsection{Permutationally invariant Bell operators}\label{sec:PIBell}

The Bell operators which are relevant to large-scale experiments inherit a strong symmetry from the way they are measured. 
In a collective-measurement Bell test on $n$ atoms one does not address the particles individually but reads out total spin components, so that the resulting Bell operator is invariant under permutations of the parties.
This motivates our focus on permutationally invariant (PI) operators.
We work in the $n$-party scenario with $2$ measurement settings and binary outcomes, the same $2$ measurements being available to every party. 
We label the input of each party, which selects the measurement operator $\mathbf{M}_{x}$, by $x \in\{0,1\}$, and the corresponding outcome, an eigenvalue of $\mathbf{M}_{x}$, by $a \in \{-1,+1\}$.
By Jordan's lemma~\cite[Lemma~2]{toner_monogamy_2006} the two observables of each party may then, without loss of generality for the extremal violation, be taken to be real $2\times2$ matrices, which we parametrize as
\begin{equation}
\mathbf{M}_x =\cos(\phi_x)\, Z+ \sin(\phi_x)\, X, \qquad x\in\{0,1\},
\label{eq:QubitMeas}
\end{equation}
where $Z$ and $X$ denote the Pauli-Z and Pauli-X matrices respectively and $\varphi_x \in [0,2 \pi)$ is the measurement angle of setting $x$, which in the rest of the paper we denote by $\varphi\coloneqq \phi_{0}$, $\theta\coloneqq \phi_{1}$.

The $k$-body PI measurement operators are built by summing a product of $k$ local measurements over all tuples of distinct qubits, indexed by a setting tuple $\vec{x}=(x_1,\dots,x_k)\in\{0,1\}^k$, whose length we write $|\vec{x}|=k$ (the body order).
More specifically, they are defined as 
\begin{equation}
     \mathbf{S}_{\vec{x}} \coloneqq \sum_{\substack{i_1,\ldots,i_k=0\\ \text{all distinct}}}^{n-1}
     \bigotimes_{j=1}^{k} \mathbf{M}_{x_j}^{i_j},
     \label{eq:kPIMeasOP}
\end{equation}
where $i_j \in \{0,\ldots,n-1\}$ denotes the qubit on which $\mathbf{M}_{x_j}^{i_j}$ acts, and the sum runs over all $k$-tuples of
distinct qubits. 
\begin{definition} A PI Bell operator of order $K$ is defined as a linear combination of PI measurement operators
\begin{equation}
\begin{aligned}
\mathbf{B} &= \sum_{\vec{x}\,\in\,\mathcal{X}_K} \alpha_{\vec{x}}\,\mathbf{S}_{\vec{x}}, \\
\mathcal{X}_K &\coloneqq \bigcup_{k=1}^K \bigl\{\vec{x}\in\{0,1\}^k : x_1\le x_2\le\cdots\le x_k\bigr\},
\end{aligned}
\label{eq:BellOp}
\end{equation}
where $\vec\alpha = (\alpha_{\vec x})_{\vec x
\in \mathcal{X}_K} \in \mathbbm{R}^{|\mathcal{X}_K|}$ are the Bell coefficients.
\end{definition}
Because the qubit sum in Eq.~\eqref{eq:kPIMeasOP} is symmetric under permuting the entries of $\vec{x}$, the operator $\mathbf{S}_{\vec{x}}$ depends only on the multiset of settings it contains.

Due to the permutational symmetry, the Bell operators can be decomposed into block-diagonal form $\mathbf{B} = \oplus_J \mathbf{B}_J$, where $\mathbf{B}_J$ denotes the subblock labelled by $J \in \{J_0, J_0 + 1, \ldots, n/2\}$~\cite{moroder_permutationally_2012}. 
Here $J_0 = 0$ if $n$ is even and $J_0 = 1/2$ if $n$ is odd. 
We call the $\mathbf{B}_J$ the $J$-th block of $\mathbf{B}$.
A basis of one representative copy of the $J$-th block is given by
\begin{equation}
\label{eq:DickeBasis}
\{\ket{D^{k}_{2J}} \otimes \ket{\psi^-}^{\otimes (n-2J)/2}\}_{k=0}^{2J},
\end{equation}
where $\ket{\psi^-} \coloneqq \frac{\ket{01}-\ket{10}}{\sqrt{2}}$ and
\begin{equation}
\label{eq:defDicke}
\ket{D^k_{2J}}\coloneqq\binom{2J}{k}^{-1/2}\sum_{\substack{w\in\{0,1\}^{2J}\\ |w|=k}}\ket{w},
\end{equation}
where the sum runs over the $\binom{2J}{k}$ computational basis states of Hamming weight $|w|=k$. 
The state $\ket{D^k_{2J}}$ is called the Dicke state of $2J$ qubits with $k$ excitations~\cite{dicke_coherence_1954,dur_three_2000}.

This symmetry is what renders the problem analytically tractable.
Each block has dimension at most $n+1$ and only polynomially many distinct blocks occur, so the problem stays tractable in the many-body limit.
The same collective structure underlies paradigmatic models such as the Lipkin--Meshkov--Glick Hamiltonian~\cite{lipkin_validity_1965,meshkov_validity_1965,glick_validity_1965} and the PI Bell operators realized in many-body experiments~\cite{tura_nonlocality_2015,schmied_bell_2016,engelsen_bell_2017}.

We refer to the largest block, $J = n/2$, as the symmetric subspace \cite{marconi_symmetric_2026}, and to an operator supported on it as a symmetric operator. We work on this block throughout. For the operators from Refs.~\cite{tura_nonlocality_2015, schmied_bell_2016}, it already contains the maximally violating state~\cite{engelsen_bell_2017}. For a general PI operator it need not in general, and one can in fact engineer the Bell operator into any block of choice as long as $K\geq 4$ and a third measurement setting along the $y$ direction is available: since $\mathbf{J}^2$ acts as $J(J+1)$ on the spin-$J$ block, one can add a penalty term $\lambda\left(\mathbf{J}^2 - J^{*}(J^{*}+1)\right)^2$, with large $\lambda > 0$, effectively raising the energy of every block except for $J=J^{*}$, therefore placing the ground state in the prescribed $J^{*}$ block. In extreme cases, \textit{i.e.}, $J^{*}=n/2$ and $J^{*} = J_0$, the linear terms $\mp \mathbf{J}^2$ are sufficient. The penalty is therefore a two-body operator in $\mathbf{J}^2$, requiring $K = 4$, and the linear terms in $\mathbf{J}^2$ just require $K=2$. We remark that $\mathbf{J}^2$ contains $\mathrm{S}_{YY}$, which is not spanned by Eq.~(\ref{eq:BellOp}) for settings in the $XZ$ plane, but it becomes available with a third setting along $Y$, which does not change the analysis below.

\subsection{The stoquasticity cone}\label{sec:stoqcone}

We can now make precise the central object of this paper. Let us fix the measurement settings $(\varphi,\theta)$ and the system size $n$, and regard the Bell coefficients $\vec{\alpha}$ as free variables in their own right, rather than as functions of $n$ or of any other parameter. 
The requirement that $\mathbf{B} = \sum_{\vec{x}\,\in\,\mathcal{X}_K} \alpha_{\vec{x}}\, \mathbf{S}_{\vec{x}}$ be stoquastic then carves out a region in coefficient space, the stoquasticity cone introduced in Ref.~\cite{li_stoquastic_2026}.
The stoquasticity conditions $(\mathbf{B}_{n/2})_{k,l} \leq 0$ for all $k \neq l$ are linear in $\vec\alpha$ and homogeneous, so the feasible set is a polyhedral cone.
\begin{definition}\label{def:stoqcone}
For fixed $(\varphi,\theta)$ and $n$, the (symmetric-block) stoquasticity cone is
\begin{equation}
\mathbbm{S}^{n,K}_{\varphi,\theta} \coloneqq
\bigl\{\vec\alpha \in \mathbbm{R}^{|\mathcal{X}_K|} \;\big|\;
(\vec{\mathbf{S}}_{k,l})^T \cdot \vec\alpha \leq 0,\;
\forall\, k \neq l\bigr\},
\label{eq:StoqCone}
\end{equation}
with 
\begin{equation}
\begin{aligned}
(\vec{\mathbf{S}}_{k,l})^T = \big( &(\mathbf{S}_0)_{k,l}, (\mathbf{S}_1)_{k,l}, (\mathbf{S}_{00})_{k,l}, 
(\mathbf{S}_{01})_{k,l}, (\mathbf{S}_{11})_{k,l},\\
&...,(\mathbf{S}_{0\mathbf{1}_{K-1}})_{k,l}, (\mathbf{S}_{\mathbf{1}_{K}})_{k,l}
\big) \nonumber
\end{aligned}
\end{equation}
where ${\mathbf{1}_{K}}$ denotes the string of $K$ ones and $\vec{\alpha}=(\alpha_{\vec{x}})_{\vec{x}\in\mathcal{X}_K}\in\mathbbm{R}^{|\mathcal{X}_K|}$ is the matching vector of Bell coefficients, ordered as in $(\vec{\mathbf{S}}_{k,l})^T$.
\end{definition}
Each pair $k\neq l$ with $|k-l|\le K$ contributes one linear inequality.
Higher off-diagonals vanish identically (Theorem~\ref{thm:BellMatElPol} below), and $(k,l)$ and $(l,k)$ coincide since the block is real symmetric (App.~\ref{app:offdiagonals}), so the cone is cut out by $O(Kn)$ inequalities, a number that grows with the system size $n$. 
For two-body operators this growth turns out to be only apparent.
Ref.~\cite{li_stoquastic_2026} showed that all but three of the inequalities are redundant, leaving a description independent of $n$.
For $K\geq 3$, however, no such finite reduction has been found. 
Already for $K=3$ each row of the first off-diagonal supports the cone at a distinct boundary point, so the number of irredundant facets grows linearly with $n$. 
Beyond economy of description, the proliferation matters practically.
Optimizing the quantum-classical gap over stoquastic operators requires the dual (vertex) description of the cone, which a facet list growing with $n$ obstructs.
A compact exact description removes this obstruction. The remainder of this work is devoted to taming this growth. First, in Sec.~\ref{sec:offdiagstruc} and Sec.~\ref{sec:reduction}, by collapsing the $O(n)$ inequalities on each off-diagonal into a single condition, and then, in Sec.~\ref{sec:sdp}, by recognizing the resulting feasible set as a projected spectrahedron.

\section{Off-diagonal element structure}\label{sec:offdiagstruc}

Throughout this work we focus on the symmetric subspace, the block $J=n/2$ (equivalently $2J=n$). Crucial to our analysis is the fact that the matrix elements of a PI measurement operator along a given off-diagonal are the values of a single explicit function of the row index $k$. 
This means that the stoquasticity conditions for a given off-diagonal can effectively be characterized by a single function. 
Additionally, PI measurement operators of order at most $K$ are real symmetric in the basis of Eq.~\eqref{eq:DickeBasis} and all matrix elements beyond the $K$-th off-diagonal vanish~\cite{tura_nonlocality_2015, li_stoquastic_2026}, as Theorem~\ref{thm:BellMatElPol} below confirms, so we only need to consider $K$ such functions. 
Their form is given by the following theorem, stated for a general off-diagonal index $d\in\{1,\ldots,K\}$.
\begin{theorem}\label{thm:BellMatElPol}
Let $\mathbf{B}_{n/2}(\varphi,\theta,\vec{\alpha})$ be a symmetric $K$-body PI Bell operator. Then $(\mathbf{B}_{n/2})_{k,k+d}=0$ for $d>K$, while for $d\in\{1,\ldots,K\}$
\begin{equation}
\label{eq:offdiagfactor}
(\mathbf{B}_{n/2})_{k,k+d}=\sqrt{f_{d}(n,k)}\;h_{d}(n,k;\varphi,\theta,\vec{\alpha}),
\end{equation}
where $f_{d}$ and $h_{d}$ are polynomials in $k$ of degree $2d$ and at most $K-d$ respectively.
\end{theorem}
The proof needs only the matrix elements of the single-body PI Pauli operators $\mathbf{S}_X = \sum_{i} X^i$, $\mathbf{S}_Y = \sum_{i} Y^i$, $\mathbf{S}_Z = \sum_{i} Z^i$ (with $i \in \{0,\ldots,n-1\}$) in the Dicke basis of Eq.~\eqref{eq:DickeBasis}~\cite{tura_nonlocality_2015}, given by
\begin{align}
\label{eq:singlePauliel}
(\mathbf{S}_X)_{k,l} &= \Gamma(l)\delta_{k,l+1} + \Gamma(k)\delta_{l,k+1}, \nonumber \\
(\mathbf{S}_Y)_{k,l} &= i\Gamma(l)\delta_{k,l+1} - i \Gamma(k)\delta_{l,k+1}, \nonumber\\
(\mathbf{S}_Z)_{k,l} &= \Xi(k) \delta_{k,l},
\end{align}
where $\Gamma(k) = \sqrt{(k+1)(n-k)}$ and $\Xi(k) = n-2k$ on the symmetric block.
On a general spin-$J$ block the same formulas hold with $n$ replaced by $2J$ (App.~\ref{app:offdiagonals}).

\begin{proof}
By Eq.~\eqref{eq:BellOp}, $\mathbf{B}$ is a linear combination of the PI measurement operators $\mathbf{S}_{\vec{x}}$, each of which decomposes as $\mathbf{S}_{\vec{x}}=\sum_{\vec{P}}c_{\vec{P}}(\varphi,\theta)\,\mathbf{S}_{\vec{P}}$ with the sum over all ordered strings of Pauli operators $X,Y,Z$, and $c_{\vec{P}}(\varphi,\theta)$ a function of $\varphi$ and $\theta$. Moreover, each $\mathbf{S}_{\vec{P}}$ is itself a sum of at most $K$-fold products $\mathbf{S}_{P_1}\cdots\mathbf{S}_{P_L}$ of single-body operators (App.~\ref{app:sub:hightosingleprod}). Every matrix element $(\mathbf{B}_{n/2})_{k,k+d}$ is therefore a linear combination of terms
\begin{equation}
\label{eq:product_of_S}
\sum_{m_1,\ldots,m_{L-1}}(\mathbf{S}_{P_1})_{k,m_1}(\mathbf{S}_{P_2})_{m_1,m_2}\cdots(\mathbf{S}_{P_L})_{m_{L-1},k+d},
\end{equation}
with $L\in\{1,\ldots,K\}$.
By Eq.~\eqref{eq:singlePauliel} each factor $\mathbf{S}_{P_{i}}$ is one step of a walk from index $k$ to $k+d$: $\mathbf{S}_X$ and $\mathbf{S}_Y$ move the index by $\pm1$, contributing $\Gamma(\min(m_{j-1},m_j))$, while $\mathbf{S}_Z$ stays put, contributing $\Xi(m_{j-1})$. An up-step from $j$ to $j+1$ and a down-step from $j+1$ to $j$ contribute the same factor $\Gamma(j)$, up to the signs and factors of $i$ carried by $\mathbf{S}_Y$, which are constants absorbed into the coefficients of $g$. To reach $k+d$ from $k$, each edge $(j,j+1)$ with $j\in\{k,\ldots,k+d-1\}$ must be crossed an odd number of times and every other edge an even number.
The single mandatory crossing of each of these $d$ edges produces the common prefactor
\begin{equation}
\prod_{j=k}^{k+d-1}\Gamma(j)=\sqrt{\,\prod_{j=k}^{k+d-1}(j+1)(n-j)\,}\eqqcolon\sqrt{f_{d}(n,k)},
\end{equation}
the square root of a polynomial $f_{d}$ of degree $2d$ in $k$. Every remaining step contributes a polynomial: an extra up--down pair on a single edge gives $\Gamma(j)^2=(j+1)(n-j)$, of degree $2$, and a step associated with $\mathbf{S}_{Z}$ gives $\Xi(j)=n-2j$, of degree $1$. The at most $K-d$ leftover steps therefore multiply the prefactor by a polynomial of degree at most $K-d$ in $k$. Therefore, from the fact that $(\mathbf{B}_{n/2})_{k,k+d}$ is a linear combination of operators $\mathbf{S}_{\vec{P}}$, we can infer Eq. \eqref{eq:offdiagfactor}.
\end{proof}
Because $\sqrt{f(n,k)}$ is a strictly positive prefactor shared by the whole $d$-th off-diagonal, the stoquasticity condition $(\mathbf{B}_{n/2})_{k,k+d}\le 0$ is equivalent to nonpositivity of $h_{d}$. For simplicity, in the rest of the paper we suppress the fixed parameters $\varphi,\theta,n$ in the notation, as the dependence of $h_{d}$ on $k$ and $\vec{\alpha}$ is the one we wish to analyze in the rest of this work.

\section{Reduction to a continuous condition}\label{sec:reduction}

The properties on the off-diagonal elements presented in Section~\ref{sec:offdiagstruc} can now be used to obtain a more compact characterization of the stoquasticity cone. 
In particular, we will show that by passing to the thermodynamic limit, we can replace the $O(n)$ conditions per off-diagonal by a small set of conditions. 
The central idea stems from the notion that as the number of hyperplanes increases, the boundary of their intersection approaches a smooth curve. Fig.~\ref{fig:threebody} illustrates this for the three-body case.
Instead of having a description in terms of the intersections, we give the characterization of the smooth surface directly, allowing for a more compact description of the boundary of the feasible set. 
This is analogous to replacing the description of the disk $\bigcap_{a^2+b^2=1}\{(u,v) : au+bv \leq r\}$ by the single nonlinear condition $\sqrt{u^2+v^2}\leq r$. 

Let us denote the set of stoquasticity conditions of the $d$-th off-diagonal as $\{h_d(k,\vec{\alpha}) \leq 0\}_k$ in $\vec{\alpha}$ parametrized by $k \in \{0,1,\ldots,n-d\}$.
In our case, we arrive at the continuous condition through solving for the value of $k$ where $h_d(k,\vec{\alpha})$ is maximized in terms of $\vec{\alpha}$ and substituting the solution back into $h_d(k,\vec{\alpha})$ to remove the $k$ dependence. 
The idea is that if $\max_k h_d(k,\vec{\alpha}) \leq 0$, then the membership condition must be satisfied for all other $k$ as well. 

However, the parameter $k \in \{0,\ldots,n-d\}$ is discrete not continuous. We solve it with the substitution $k = (n-d)x$, the index $k \in \{0,\ldots,n-d\}$ gets replaced by the bounded continuous variable $x \in [0,1]$, whose endpoints coincide with the first and last rows, and the discrete supremum $\sup_k$ becomes $\sup_{x\in[0,1]}$.
For convenience, we can divide the polynomial we obtain after this substitution by $(n-d)^{K-d}$, as the stoquasticity constraints are invariant under division by a positive real number.
Then, we get the following sufficient stoquasticity condition for the $d$-th off-diagonal
\begin{equation}
\sup_{x\in[0,1]}\,\mathcal{h}_{d}(x,\vec{\alpha})\leq 0,
\label{eq:StoqSup}
\end{equation}
where $\mathcal{h}_{d}(x,\vec{\alpha}) \coloneqq (n-d)^{-(K-d)}\,h_{d}(k,\vec{\alpha})\big|_{k=(n-d)x}$. The condition is only sufficient, as the existence of $x$ such that $\mathcal{h}_{d}(x,\vec{\alpha})>0$ does not necessarily imply $h(k,\vec{\alpha})>0$ for some $k$. However, \eqref{eq:StoqSup} is necessary and sufficient for stoquasticity in the thermodynamic limit $n\rightarrow \infty$, as $k/(n-d)$ becomes dense in $[0,1]$. Moreover, the normalization we have imposed on $\mathcal{h}_{d}$ ensures that no monomial term diverges in the thermodynamic limit, which allows us to work with the thermodynamic limit polynomial, denoted as $\mathcal{h}^{\infty}_{d}(x,\vec{\alpha})\coloneqq\lim_{n\to\infty}\mathcal{h}_{d}(x,\vec{\alpha})$.

It is worth being explicit about the scope of this limit, since the whole reduction rests on it. The Bell coefficients $\vec{\alpha}$ are held fixed as $n\to\infty$.
A $\kappa$-body measurement operator contributes to the $d$-th off-diagonal a polynomial of degree at most $\kappa-d$. After the substitution $k=(n-d)x$ and division by $(n-d)^{K-d}$ its contribution carries a factor $(n-d)^{-(K-\kappa)}$, so at fixed $\vec{\alpha}$ every $\kappa$ apart from the top one is suppressed, and the limit $\mathcal{h}^{\infty}_{d}(x,\vec{\alpha})$ keeps only the highest-order coefficients.

For the three-body operator of Sec.~\ref{sec:three-body} this is why the interior geometry is governed by the three-body coefficients alone. The situation changes for \emph{families} of operators whose coefficients scale with the system size, as arises for some many-body Bell inequalities~\cite{tura_nonlocality_2015}. If a lower-body coefficient is allowed to grow like $n^{K-\kappa}$, its contribution is no longer suppressed and re-enters the limit. The reduction still produces the nonnegativity of a univariate polynomial of degree at most $K-d$ on $[0,1]$, so the semidefinite characterization of the next section applies unchanged.
However, the limit polynomial is then a different one, and which body orders survive must be tracked through the chosen scaling.

A related subtlety lives at the endpoints $x=0,1$. 
The square-root factors $\sqrt{x(1-x)}$ that make these points delicate do not, however, enter the condition.
Instead, they belong to the common non-negative prefactor $\Gamma(k)\cdots\Gamma(k+d-1)$ that has been divided out, leaving Eq.~\eqref{eq:StoqSup} a condition on the smooth polynomial alone.
For fixed $\vec{\alpha}$ the supremum is therefore attained either in the bulk at some $x^{\ast}\in(0,1)$ or at an endpoint, and the endpoint conditions $\mathcal{h}_{d}(0,\vec{\alpha})\le0$ and $\mathcal{h}_{d}(1,\vec{\alpha})\le0$ are the lattice rows $k=0$ and $k=n-d$ themselves, i.e. they are exact at every $n$. 

The supremum of $\mathcal{h}_{d}(x,\vec{\alpha})$ is attained either at an interior stationary point $x_j^{*}(\vec{\alpha})$ satisfying
\begin{equation}
\label{eq:stationarity}
    \partial_x \mathcal{h}_{d}(x_j^{*},\vec{\alpha}) = 0, \qquad
    \partial_x^2 \mathcal{h}_{d}(x_j^{*},\vec{\alpha}) < 0,
\end{equation}
with $x_j^{*}(\vec{\alpha})\in (0,1)$, or at an endpoint $x \in \{0,1\}$. 
In the degenerate cases where $\partial_x^2 \mathcal{h}_{d}(x_j^{*},\vec{\alpha}) = 0$, higher-order derivative tests are needed to find the local maxima. 
In particular, $x^{*}(\vec{\alpha})$ is a local maximum if and only if the first non-zero derivative occurs at an even order and is negative. 
The full stoquasticity condition for the $d$-th off-diagonal is then
\begin{equation}
\label{eq:elimCondition}
    \max\bigl\{\mathcal{h}_{d}(0,\vec{\alpha}),\; \mathcal{h}_{d}(1,\vec{\alpha}),\;
    \max_j \mathcal{h}_{d}(x_j^{*}(\vec{\alpha}),\vec{\alpha})\bigr\} \leq 0,
\end{equation}
where $j$ enumerates all $x_j^{*}(\vec{\alpha})$ that are interior local maxima.
As $\mathcal{h}_{d}(x,\vec{\alpha})$ is a polynomial of degree at most $K-d$ in $x$, there are at most $\lceil(K-d-1)/2\rceil$ such solutions $x_j^{*}(\vec{\alpha})$.
We note that the maximizer $x^{*}(\vec{\alpha})\in (0,1)$ is always algebraic, being a root of the polynomial $\partial_x \mathcal{h}_{d}$, but is not in general expressible in terms of radicals. 
In particular, for polynomials of degree five or higher, no closed-form solution in radicals exists in general.
Therefore, it is not guaranteed that for Bell operators of arbitrary order $K$, such a closed-form radical expression exists. 
In Sec.~\ref{sec:sdp} we first show that the collapsed condition~\eqref{eq:StoqSup} admits an exact semidefinite formulation valid for every order $K$.
We then apply the approach to the two-body and three-body stoquasticity cones explicitly in Sec.~\ref{sec:two-body} and~\ref{sec:three-body}, which are its lowest-degree instances. 

We close this section with a structural feature of the limit polynomials.
\begin{lemma}[Parity of the limit polynomials]\label{lem:parity}
For every $d$, the limit polynomial $\mathcal{h}^{\infty}_{d}(x,\vec{\alpha})$ has definite parity under the flip $x\mapsto1-x$:
\begin{equation}
\mathcal{h}^{\infty}_{d}(1-x,\vec{\alpha})=(-1)^{K-d}\,\mathcal{h}^{\infty}_{d}(x,\vec{\alpha}).
\end{equation}
\end{lemma}
\begin{proof}
Let us first consider a walk associated with Eq. \eqref{eq:product_of_S}. After removing the common square root prefactor, such a walk consists of $z$ steps from $k$ to $k$, each contributing $\Xi\to n(1-2x)$, and $m$ up--down step pairs, each contributing $\Gamma^{2}\to n^{2}\,x(1-x)$, with $z+2m=K-d$. Notice that the former contribution is odd in $x$ whereas the latter is even, and so the product of them corresponds to a path associated with Eq. \eqref{eq:product_of_S} that is either odd or even, depending on the parity of $z$. Since we have $z=K-d-2m$ the parity of $z$ equals the parity of $K-d$, meaning that all possible walks from $k$ to $k+d$ will yield functions of parity $K-d$. As a linear combination of even (odd) functions is itself even (odd), the parity of $\mathcal{h}^{\infty}_{d}$ equals the parity of $K-d$, concluding the proof.
\end{proof}
The parity has three consequences. First, for odd $K-d$ the limit polynomial is odd about $x=\tfrac12$. Because the reflection $x\mapsto1-x$ maps $[0,1]$ onto itself, a nonzero such polynomial is positive somewhere on $[0,1]$, so $\sup_{x\in[0,1]}\mathcal{h}^{\infty}_{d}\le0$ can hold only if it vanishes identically.
Its coefficients, linear in $\vec{\alpha}$, then all vanish, and the $d$-th off-diagonal degenerates from an inequality to a set of \emph{equality} constraints in the strict limit. Second, for even $K-d=2\ell$ the polynomial is even about $x=\tfrac12$, so in the variable $u=(2x-1)^{2}\in[0,1]$ it is a polynomial of degree $\ell$, and stoquasticity asks it to be nonnegative on $[0,1]$.Third, the parity has a physical origin: $x\mapsto1-x$ is the spin flip $k\mapsto n-k$ about the equator, a symmetry of the leading-order collective operators that the distinct-index corrections break only at $O(1/n)$. At the orders whose strict limit is polyhedral, the boundary therefore bends only at $O(1/n)$.

\section{Semidefinite characterization of the feasibility region}\label{sec:sdp}

The collapse of Sec.~\ref{sec:reduction} replaced the $O(n)$ stoquasticity hyperplanes of the $d$-th off-diagonal by the single condition $\sup_{x\in[0,1]}\mathcal{h}_{d}(x,\vec{\alpha})\le 0$ of Eq.~\eqref{eq:StoqSup}.
By Theorem~\ref{thm:BellMatElPol}, $\mathcal{h}_{d}(x,\vec{\alpha})$ is a univariate polynomial of degree at most $K-d$ in $x$ whose coefficients are linear in $\vec{\alpha}$.
The condition is therefore an instance of \emph{nonnegativity of a univariate polynomial on an interval}, a property that is exactly representable by a linear matrix inequality (LMI).
In this section we make this representation explicit.
Its consequence is that the entire thermodynamic-limit feasible set is a projected spectrahedron (a spectrahedral shadow).
Here, membership is decided by a single semidefinite program (SDP), and the two- and three-body cones of Secs.~\ref{sec:two-body}--\ref{sec:three-body} are its lowest-degree instances.
As the constraints are univariate, the description is exact, requiring no relaxation hierarchy.

We first name the object of interest, built on the limit polynomials $\mathcal{h}_{d}$.
\begin{definition}[Feasibility region]\label{def:thermoregion}
For fixed order $K$, particle number $n$, and settings $(\varphi,\theta)$, where $(\varphi,\theta)$ are suppressed from the notation, the feasibility region is
\begin{widetext}
\begin{equation}
\label{eq:FeasRegApprox}
\widetilde{\mathbbm{S}}^{K}_{n}\coloneqq
\Bigl\{\vec{\alpha}\in\mathbbm{R}^{|\mathcal{X}_K|} \;\Big|\;
\sup_{x\in[0,1]}\mathcal{h}_{d}(x,\vec{\alpha})\le 0,\ d=1,\dots,K\Bigr\}
=\bigcap_{d=1}^{K}\bigl\{\vec{\alpha}\mid -\mathcal{h}_{d}(\cdot,\vec{\alpha})\ge 0\ \text{on}\ [0,1]\bigr\}.
\end{equation}
\end{widetext}
\end{definition}
If we take the thermodynamic limit, we replace $\mathcal{h}_{d}$ by $\mathcal{h}^{\infty}_{d}$ in Eq.~\eqref{eq:FeasRegApprox}, and denote it as $\widetilde{\mathbbm{S}}^{K}_{\infty}$, which we call the thermodynamic-limit feasibility region. For each fixed $x$, $\{\vec{\alpha}\mid \mathcal{h}^{\infty}_{d}(x,\vec{\alpha})\le 0\}$ is a half-space, so $\mathbbm{S}^{K}_{\infty}$ is an intersection of half-spaces and hence convex.
Since the matrix elements are linear and homogeneous in $\vec{\alpha}$, it is in fact a closed convex cone.
The same construction applies at finite $n$: since the lattice points $k/(n-d)$ lie in $[0,1]$, with the endpoints among them, the condition $\sup_{x\in[0,1]}\mathcal{h}_{d}(x,\vec{\alpha})\le0$ with the finite-$n$ coefficients defines a closed convex cone contained in $\mathbbm{S}^{n,K}_{\varphi,\theta}$, so the semidefinite program below provides an exact stoquasticity certificate at every finite $n$. 

It is important to stress that while the characterization of polynomials by an SDP is exact for both $\mathcal{h}_{d}$ and $\mathcal{h}_{d}^{\infty}$, it is only an exact representation of a stoquasticity cone for the latter. This is due to the fact that, as mentioned before, Eq. \eqref{eq:StoqSup} is only a sufficient condition for stoquasticity, and so $\widetilde{\mathbb{S}}_{n}^{K}$ is an inner approximation of $\mathbbm{S}^{n,K}_{\varphi,\theta} $.

\subsection{Sum-of-squares certificate}\label{sec:sdp:sos}

Let us begin with a remark that everything in this section applies verbatim to $\mathcal{h}_{d}$ and $\mathcal{h}^{\infty}_{d}$ alike. However, for simplicity, we only write $\mathcal{h}_{d}$ below.

Nonnegativity of a univariate polynomial on $[0,1]$ is governed by the Markov--Luk\'acs theorem~\cite{powers_polynomials_2000, blekherman_semidefinite_2012}, which we state in the form convenient here as a lemma.
Let $z_r(x)\coloneqq(1,x,\dots,x^{r})^{\top}$ denote the vector of monomials up to degree $r$.

\begin{lemma}[LMI representation of an off-diagonal constraint]\label{thm:sos}
Let $p(x)\coloneqq -\mathcal{h}_{d}(x,\vec{\alpha})$, a polynomial of degree at most $K-d$.
Then $p(x)\ge 0$ for all $x\in[0,1]$ if and only if there exist positive semidefinite Gram matrices realizing the weighted sum-of-squares decomposition
\begin{widetext}
\begin{equation}
\label{eq:sosrep}
p(x)=
\begin{cases}
z_{\ell}(x)^{\top}G_0\,z_{\ell}(x)+x(1-x)\,z_{\ell-1}(x)^{\top}G_1\,z_{\ell-1}(x), & K-d=2\ell,\\[6pt]
x\,z_{\ell}(x)^{\top}G_0\,z_{\ell}(x)+(1-x)\,z_{\ell}(x)^{\top}G_1\,z_{\ell}(x), & K-d=2\ell+1,
\end{cases}
\end{equation}
\end{widetext}
with $G_0\succeq 0$ of size $(\ell{+}1)\times(\ell{+}1)$, $G_1\succeq 0$ of size $\ell\times\ell$ (even case) or $(\ell{+}1)\times(\ell{+}1)$ (odd case).
For $\ell=0$ the second term of the even case is absent.
\end{lemma}
Equating the coefficients of equal powers of $x$ on the two sides of Eq.~\eqref{eq:sosrep} produces $K-d+1$ linear equations relating the entries of the Gram matrices to the coefficients of $p$.
Because those coefficients are linear in $\vec{\alpha}$ (Theorem~\ref{thm:BellMatElPol}), these equations are jointly linear in $(\vec{\alpha},G_0,G_1)$.
The $d$-th stoquasticity condition therefore carves out a spectrahedron in the lifted variables $(\vec{\alpha},G_0,G_1)$, and its projection onto $\vec{\alpha}$:
\begin{equation}
\widetilde{\mathbbm{S}}^{(d)}_n\coloneqq\bigl\{\vec{\alpha}\mid \exists\,G_{0/1}\succeq 0\ \text{satisfying Eq.~\eqref{eq:sosrep}}\bigr\},
\end{equation}
is a spectrahedral shadow~\cite{ben-tal_lectures_2001, blekherman_semidefinite_2012}.
The full feasibility region is the intersection $\widetilde{\mathbbm{S}}^{K}_{n}=\bigcap_{d=1}^{K}\widetilde{\mathbbm{S}}^{(d)}_n$, which is again SDP-representable.
\begin{corollary}\label{cor:membership}
Membership $\vec{\alpha}\in \widetilde{\mathbbm{S}}^{K}_{n}$ is decided by a single semidefinite feasibility program whose blocks have size at most $\lfloor(K-d)/2\rfloor+1$ on the $d$-th off-diagonal, hence at most $\lceil K/2\rceil$ overall.
A feasible collection
$\{G_0^{(d)},G_1^{(d)}\}_{d=1}^{K}$ therefore provides an explicit
certificate that $\vec{\alpha}\in \widetilde{\mathbbm{S}}^{K}_{n}$.
Conversely, if $\vec{\alpha}\notin \widetilde{\mathbbm{S}}^{K}_{n}$, then
there exist $d\in\{1,\ldots,K\}$ and $x^\star\in[0,1]$ such that
$\mathcal{h}_d(x^\star,\vec{\alpha})>0$.
\end{corollary}
The decisive feature is exactness. For univariate polynomials, Eq.~\eqref{eq:sosrep} is necessary as well as sufficient. 
Every polynomial that is nonnegative on $[0,1]$ admits a decomposition of the form~\eqref{eq:sosrep} using Gram matrices of the finite sizes stated above~\cite{powers_polynomials_2000,nesterov_squared_2000}. Hence the resulting semidefinite program characterizes nonnegativity on $[0,1]$ exactly, rather than providing only a relaxation of it.

This is unlike the general multivariate situation, where a Positivstellensatz certificate may require sum-of-squares multipliers of unbounded degree, organized into a hierarchy of semidefinite programs that need not terminate~\cite{lasserre_global_2001,laurent_sums_2008}.
The off-diagonal structure established in Sec.~\ref{sec:offdiagstruc}, which compresses each off-diagonal to a single continuous index $x$, places the stoquasticity problem in this exact, univariate regime.

Moreover, notice that by the parity of Lemma~\ref{lem:parity}, the even off-diagonals land, as $n\to\infty$, in a symmetric slice on which the body stays polyhedral through $K-d=2$. The three-body quadric of Fig.~\ref{fig:threebody} is therefore a finite-$n$ feature, and a body's curvature first survives the limit at $K-d=4$. The same construction applies at every order $K$, only the degrees $K-d$ changing. 

For $K\ge7$ the interior maximizer $x^{\star}(\vec{\alpha})$ generally cannot be written in radicals (Sec.~\ref{sec:reduction}).
Lemma~\ref{thm:sos} bypasses this obstruction, since membership is decided without ever locating $x^{\star}$.
The semidefinite program in the thermodynamic limit thus provides for arbitrary $K$ what explicit elimination delivers only through $K=3$, characterizing the whole region, interior and endpoint regimes at once. On the subregion $\partial'\mathbbm{S}$, App.~\ref{app:boundary} describes the boundary locally: around a generic $\vec\alpha$, one at which the polynomial signs in Sturm's theorem do not vanish, the roots of $\partial_x h_d$ move continuously with $\vec\alpha$, so the same case applies on a whole open ball centred at $\vec\alpha$, and the boundary of $\partial' S$ is given by the geometric loci where such a sign crosses zero.

\subsection{Optimization over the stoquasticity cone}\label{sec:sdp:opt}

Because $\widetilde{\mathbbm{S}}^{K}_{\infty}$ is SDP-representable, optimizing any linear functional of the Bell coefficients over it is itself a semidefinite program.
Designing extremal stoquastic operators, or maximizing a linear figure of merit subject to stoquasticity, takes the form
\begin{equation}
\label{eq:sdpopt}
\begin{aligned}
\max_{\vec{\alpha}}\ &\vec{w}^{\top}\vec{\alpha} \\
\text{s.t.}\quad &\vec{\alpha}\in\widetilde{\mathbbm{S}}^{K}_{\infty},\quad \lVert\vec{\alpha}\rVert_2\le1,
\end{aligned}
\end{equation}
a semidefinite program in $\vec{\alpha}$ and the Gram matrices of Lemma~\ref{thm:sos}, with the second-order-cone norm bound (itself LMI-representable) fixing the scale of the homogeneous cone.
Stoquasticity can in this way be imposed jointly with the other semidefinite constraints that arise in the permutationally invariant Bell setting, and the semidefinite description already lets us design extremal Bell operators, whose vertex characterization we leave to future work.

\section{Two-body case}\label{sec:two-body}

The two-body case serves as the first nontrivial application of the approach introduced in Sec.~\ref{sec:reduction}.
It is also of direct physical relevance: two-body PI Bell operators underlie the largest Bell violations observed experimentally to date~\cite{schmied_bell_2016, engelsen_bell_2017}, and their stoquasticity structure was the original motivation for introducing the stoquasticity cone in Ref.~\cite{li_stoquastic_2026}.
For $K=2$, the first off-diagonal is at most linear in $k$ and the second off-diagonal is constant in $k$.
The stoquasticity condition $h_1(k,\vec{\alpha})\le0$ of the first off-diagonal reads explicitly~\cite{li_stoquastic_2026} 
\begin{widetext}
\begin{equation}
\label{eq:twobodyfirstdiag}
\alpha_0 \sin (\varphi) + \alpha_1 \sin (\theta) + \alpha_2 (-2 k+n-1) \sin (2 \varphi ) + \alpha_3 (-2 k+n-1) \sin ( \varphi + \theta ) + \alpha_4 (-2 k+n-1) \sin (2 \theta ) \leq 0.
\end{equation}
\end{widetext}
and the stoquasticity condition of the second off-diagonal, independent of $k$, is~\cite{li_stoquastic_2026} 
\begin{equation}
 \alpha_2 \sin ^2(\varphi )+ \alpha_3 \sin ( \varphi) \sin( \theta ) + \alpha_4 \sin ^2(\theta ) \leq 0.
 \nonumber
\end{equation}
After the substitution $k = (n-1) x$ and division by $n-1$, the single-body terms of Eq.~\eqref{eq:twobodyfirstdiag} are suppressed by $(n-1)^{-1}$ while $(-2k+n-1)=(n-1)(1-2x)$ exactly.
In the thermodynamic limit Eq.~\eqref{eq:twobodyfirstdiag} thus becomes
\begin{equation}
(1 - 2x)\bigl[\alpha_2 \sin (2 \varphi ) + \alpha_3 \sin ( \varphi + \theta ) + \alpha_4 \sin (2 \theta )\bigr] \leq 0.
\nonumber
\end{equation}
As the first off-diagonal is at most a degree-one polynomial in $x$, at every finite $n$ as well as in the limit, the supremum can only be attained at one of the two boundaries $x=0$ or $x=1$, which are the rows $k=0$ and $k=n-1$ themselves (with the single-body terms retained at finite $n$), resulting in two hyperplanes whose irredundancy is shown in Ref.~\cite{li_stoquastic_2026}. 
The second off-diagonal is independent of $x$ and the supremum is thus the same across the entire interval $[0,1]$.
This recovers the three-hyperplane characterization of $\mathbbm{S}^{n,2}_{\varphi,\theta}$ of Ref.~\cite{li_stoquastic_2026} exactly, at every $n$, directly from the elimination procedure.
Furthermore, the candidate maximizers are the two endpoints irrespective of $\vec{\alpha}$, so the three hyperplanes are sufficient for all $\vec{\alpha}$.
Equivalently, writing $A\coloneqq\alpha_0\sin\varphi+\alpha_1\sin\theta$ for the single-body part and $s\coloneqq\alpha_2\sin2\varphi+\alpha_3\sin(\varphi+\theta)+\alpha_4\sin2\theta$ for the two-body one, the two endpoint rows are $A\pm(n-1)s\le0$, that is $|s|\le -A/(n-1)$: the single-body term leaves $s$ free within a window of width $O(1/n)$. Therefore, the cone is full-dimensional at every finite $n$, i.e., it is a five-dimensional cone in $\mathbbm{R}^5$ with these three facets. 

In the strict thermodynamic limit the two conditions become $\pm s\le0$ and pin $s=0$, so the limit cone $\widetilde{\mathbbm{S}}^{2}_{\infty}$ degenerates to $\{s=0,\ \alpha_2\sin^2\varphi+\alpha_3\sin\varphi\sin\theta+\alpha_4\sin^2\theta\le0\}$, a four-dimensional cone with $\alpha_0,\alpha_1$ now free, and the single-body coefficients drop out altogether. 
The exact finite-$n$ hyperplanes above, which retain the single-body data, are therefore the physically relevant description. A related strict-limit degeneration occurs for three-body operators in Sec.~\ref{sec:three-body}, where the interior conic flattens to a pair of hyperplanes.

\section{Three-body case}\label{sec:three-body}

The three-body case is qualitatively richer than the two-body case.
For two-body operators, the constraint functions $\mathcal{h}_{d}(x, \vec{\alpha})$ were at most linear in $x$, so their suprema over $[0,1]$ were always attained at the endpoints $x = 0$ or $x = 1$.
For three-body operators, this is similarly the case for the second $(d=2)$ and third $(d=3)$ off-diagonals, as $\mathcal{h}_2(x, \vec{\alpha})$ and $\mathcal{h}_3(x, \vec{\alpha})$ are at most linear and constant in $x$, respectively. 
However, the first off-diagonal $(d=1)$ constraint is a degree-two
polynomial in $x$, meaning the supremum can now be attained at an
interior point $x^* \in (0,1)$ depending on the value of $\vec{\alpha}$.
Therefore, $x^*$ itself is a function of the Bell coefficients.
In particular, as we will show, the geometric object characterizing the boundary of the stoquasticity cone in this interior regime is a quadratic form. 
This means that it can be written in the form $\vec{\alpha}^T \mathbf{Q} \vec{\alpha}$, where $\mathbf{Q} \in \mathbbm{R}^{9 \times 9}$ is a symmetric matrix.
To see this, we first note that the function for the first off-diagonal of the three-body PI Bell operator is of the form $h_1(k, \vec{\alpha}) = k^2 \vec{A}' \cdot \vec{\alpha} + k \vec{B}' \cdot \vec{\alpha} + \vec{C}' \cdot \vec{\alpha}$
as is given by Theorem~\ref{thm:BellMatElPol}.
Here $\vec{A}', \vec{B}', \vec{C}'$ are implicitly functions of $\varphi, \theta, n$. 
The explicit expressions of the off-diagonal elements of single-body, two-body and three-body PI measurement operators can be found in App.~\ref{app:offdiagonals}. 
After making the substitution $k=(n-1)x$ and dividing by $(n-1)^2$, the general form remains the same
\begin{equation}
\label{eq:ThreeBodyFirstDiag}
    \mathcal{h}_1(x, \vec{\alpha}) = x^2 \vec{A} \cdot \vec{\alpha} + x \vec{B} \cdot \vec{\alpha} + \vec{C} \cdot \vec{\alpha}
\end{equation}
with $\vec{A}=\vec{A}'$, $\vec{B}=\vec{B}'/(n-1)$ and $\vec{C}=\vec{C}'/(n-1)^2$.
Theorem~\ref{thm:QuadraticSurface} shows how in the thermodynamic limit, the stoquasticity condition on the first off-diagonal can be characterized as a quadratic form in terms of $\vec{A},\vec{B},\vec{C}$. 
\begin{theorem}
\label{thm:QuadraticSurface}
The stoquasticity condition of the three-body PI Bell operator where the supremum of $\mathcal{h}_1(x,\vec{\alpha})$ is attained in the interior $x^* \in (0,1)$ is of the form 
\begin{equation}
\label{eq:QForm}
\begin{aligned}
- 4(\vec{A}\cdot\vec{\alpha})\, \mathcal{h}_1(x^*,\vec{\alpha}) &= \vec{\alpha}^{\top} \mathbf{Q}\,\vec{\alpha} \leq 0,\\
\text{with}\quad \mathbf{Q} &= \left[\vec{B}\vec{B}^{\top}-2\left(\vec{A}\vec{C}^{\top} + \vec{C}\vec{A}^{\top}\right)\right].
\end{aligned}
\end{equation}
\end{theorem}
\begin{proof}
As $\mathcal{h}_1(x, \vec{\alpha})$ is a quadratic function in $x$, the first derivative with respect to $x$ is linear and admits a unique root of the form $x^* = - \frac{ \vec{B} \cdot \vec{\alpha}}{ 2 \vec{A} \cdot \vec{\alpha}}$.
Substituting the solution for $x^*$
into $\mathcal{h}_1(x,\vec{\alpha})$ gives $\mathcal{h}_1(x^*, \vec{\alpha}) =(\vec{C}\cdot\vec{\alpha}) - (\vec{B}\cdot\vec{\alpha})^2/ 4(\vec{A}\cdot\vec{\alpha})$.
Multiplying $\mathcal{h}_1(x^*, \vec{\alpha})$ by $- 4(\vec{A}\cdot\vec{\alpha})$, choosing the sign such that we multiply by a positive factor, allows us to remove the denominator.
This returns the stoquasticity condition
\begin{equation}
\label{eq:StoqQuad}
\begin{aligned}
&- 4(\vec{A}\cdot\vec{\alpha})(\vec{C}\cdot\vec{\alpha}) + (\vec{B}\cdot\vec{\alpha})^2 \leq 0\\
&\quad \text{if} \quad - 4(\vec{A}\cdot\vec{\alpha}) > 0.
\end{aligned}
\end{equation}
The condition is always satisfied as an interior supremum of the quadratic $\mathcal{h}_1$ requires concavity, $\partial_x^2 \mathcal{h}_1=2\,\vec{A}\cdot\vec{\alpha}<0$ (for $\vec{A}\cdot\vec{\alpha}\ge0$ the supremum is attained at an endpoint), so $- 4(\vec{A}\cdot\vec{\alpha}) > 0$.

We note that $(\vec{A}\cdot\vec{\alpha})(\vec{C}\cdot\vec{\alpha}) 
= \vec{\alpha}^{\top}
\tfrac{\vec{A}\vec{C}^{\top}+\vec{C}\vec{A}^{\top}}{2}
\vec{\alpha}$
is a quadratic form in $\vec{\alpha}$, and so is $(\vec{B}\cdot\vec{\alpha})^2$. 
Therefore $-4(\vec{A}\cdot\vec{\alpha}) \mathcal{h}_1(x^*, \vec{\alpha})$ is a quadratic form as well and substituting the expressions $(\vec{A}\cdot\vec{\alpha})(\vec{C}\cdot\vec{\alpha}) 
= \vec{\alpha}^{\top}
\tfrac{\vec{A}\vec{C}^{\top}+\vec{C}\vec{A}^{\top}}{2}
\vec{\alpha}$ and $\vec{\alpha}^{\top}\vec{B} \vec{B}^{\top} \vec{\alpha}$ into Eq.~\eqref{eq:StoqQuad} gives the desired result.
\end{proof}

In the form of Lemma~\ref{thm:sos}, Eq.~\eqref{eq:sosrep} reads $p(x)=z_1(x)^{\top}G_0\,z_1(x)+x(1-x)\,G_1$ with a $2\times2$ matrix $G_0\succeq0$ and a scalar $G_1\ge0$.
Eliminating the auxiliary variable $G_1$ projects this spectrahedron onto exactly the feasibility conditions of Sec.~\ref{sec:three-body}: the endpoint inequalities $\mathcal{h}_1(0,\vec{\alpha})\le0$ and $\mathcal{h}_1(1,\vec{\alpha})\le0$, supplemented in the interior regime by the quadratic condition $\vec{\alpha}^{\top}\mathbf{Q}\vec{\alpha}\le0$, which is binding precisely when the supremum is attained in $(0,1)$.
In the interior regime the conic condition coincides, up to the positive factor $4$, with $\det G_0\ge0$ evaluated at $G_1=0$, and the case distinction between interior and endpoint suprema together with the sign choice $\pm$ is absorbed into the single semidefinite condition $G_0\succeq0,\,G_1\ge0$.
The explicit coefficient matching is given in App.~\ref{app:sdpgram}, and App.~\ref{app:lift} assembles the full $K=3$ description into a single canonical LMI with one auxiliary variable, showing moreover that this variable cannot be eliminated: the three-body cone is a spectrahedral shadow but not a spectrahedron.
Geometrically, the three-body conic boundary of Sec.~\ref{sec:three-body} is the boundary of the projection of this lifted spectrahedron onto the Bell coefficients.

We now take a closer look at the form of $\mathbf{Q}$ in the thermodynamic limit and find that it simplifies.
Notice that $\vec{A}',\vec{B}',\vec{C}'$ are of at most degree $0$, degree $1$ and degree $2$ in $n$ respectively.
Consequently, $\vec{A},\vec{B},\vec{C}$ are of at most degree $0$, degree $1$ and degree $2$ in $(n-1)^{-1}$, and $\mathbf{Q}$ decomposes into three $n$-independent terms, $\mathbf{Q} =  \left( \mathbf{Q}_1 + \frac{1}{n-1} \mathbf{Q}_2 + \frac{1}{(n-1)^2} \mathbf{Q}_3 \right)$.
The decomposition is exact and fully explicit, as a direct computation from the off-diagonal formulas as shown in added separate appendix with this calculation, see App.~\ref{app:Qdecomp}. With $\vec{A}=\vec{A}_1$ at every $n$, one has
\begin{widetext}
\begin{align}
\vec{B} &= -\vec{A}_1-\tfrac{2}{n-1}\,\vec{e}_2,\qquad \vec{C} = \Bigl(1-\tfrac{1}{n-1}\Bigr)\vec{C}_1+\tfrac{1}{n-1}\,\vec{e}_2+\tfrac{1}{(n-1)^{2}}\,\vec{e}_1, \label{eq:ABCclosed}\\
\mathbf{Q}_1 &=  \vec{A}_1\vec{A}_1^{\top}-2\bigl(\vec{A}_1\vec{C}_1^{\top} + \vec{C}_1\vec{A}_1^{\top}\bigr) ,\qquad
\mathbf{Q}_2 = 2\bigl(\vec{A}_1\vec{C}_1^{\top} + \vec{C}_1\vec{A}_1^{\top}\bigr),\qquad
\mathbf{Q}_3 &= 4\vec{e}_2\vec{e}_2^{\top} - 2\left( \vec{A}_1\vec{e}_1^{\top} + \vec{e}_1\vec{A}_1^{\top} \right), \label{eq:Qclosed}
\end{align}
\end{widetext}
where $\vec{A}_1,\vec{C}_1$ denote the $n\to\infty$ limits of $\vec{A},\vec{C}$, supported on the three-body coefficients and given explicitly in App.~\ref{app:ThreeBodyHPlanes}, and $\vec{e}_1$, $\vec{e}_2$ are the single- and two-body coefficient vectors with nonvanishing entries $(\sin\varphi,\sin\theta)$ and $\bigl(\sin2\varphi,\sin(\varphi{+}\theta),\sin2\theta\bigr)$, respectively.
Notably, the two-body data cancels from $\mathbf{Q}_2$, whose sole effect is to soften the conic part of $\mathbf{Q}_1$.
The single- and two-body coefficients enter only at second order, through $\mathbf{Q}_3$.
The leading-order term $\mathbf{Q}_1$ has rank $2$ for all non-degenerate settings, as made concrete in Theorem~\ref{thm:rank2} below. Recall that the \emph{signature} $(p,q)$ of a real symmetric matrix is its number of positive and negative eigenvalues.

\begin{figure*}[t]
\centering
\includegraphics[width=0.86\textwidth]{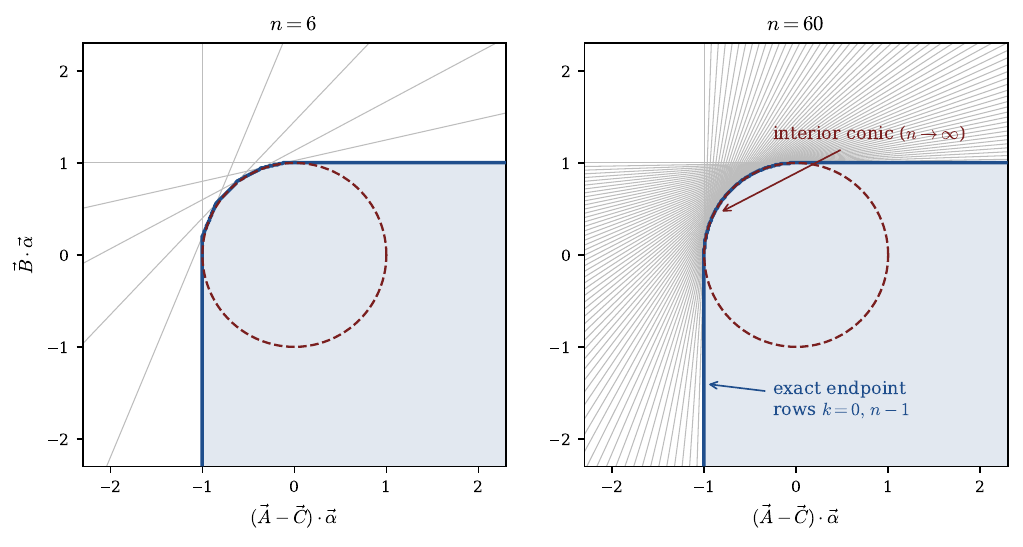}
\caption{\textbf{Geometry of the three-body interior stoquasticity condition.} The first off-diagonal is the quadratic $\mathcal{h}_1(x)=ax^{2}+bx+c$, with $a=\vec{A}\cdot\vec{\alpha}$, $b=\vec{B}\cdot\vec{\alpha}$, $c=\vec{C}\cdot\vec{\alpha}$. We plot the slice $a+c=-1$ in the axes $s=a-c$ and $b$. Each row $k$ contributes one stoquasticity constraint drawn as a grey line tangent to the interior conic $b^{2}\le4ac$, which on this slice is the circle $s^{2}+b^{2}=1$ (dashed). Off this arc the supremum is attained at an endpoint, where the exact rows $s=-1$ ($k{=}0$) and $b=1$ ($k{=}n{-}1$) take over. The shaded feasible set is bounded by these two edges together with the conic arc. As $n$ grows from $6$ to $60$ the rows fill in, and the boundary collapses uniformly onto the conic arc, reducing the $O(n)$ constraints to a single condition.}
\label{fig:threebody}
\end{figure*}

\begin{theorem}
\label{thm:rank2}
    The matrix $\mathbf{Q}$ has rank at most $3$ with the leading-order term $\mathbf{Q}_{1}$ of rank at most $2$. Moreover, for $\theta\neq \varphi$, $\mathbf{Q}_{1}$ has rank $2$, is indefinite, and has a signature $(1,1)$.
\end{theorem}
\begin{proof}
    The form of $\mathbf{Q}$ given by Eq.~\eqref{eq:QForm} shows that it is of at most rank $3$, as it is characterized by three vectors $\vec{A},\vec{B},\vec{C}$.
    To see that $\mathbf{Q}_1$ is rank $2$, we first recall that the first off-diagonal function $h_1(k,\vec{\alpha}) = \sum_{|\vec{x}| \leq 3} \alpha_{\vec{x}}\,\mathbf{S}_{\vec{x}}^{k,k+1}/\Gamma(k)$ is a linear combination of the first off-diagonal elements of PI measurement operators up to order 3.
    The matrix elements of the three-body PI measurement operators are of the form
    \begin{equation}
        \mathbf{S}_{|\vec{x}|=3}^{k,k+1}/\Gamma(k) = p(\varphi,\theta) F(k,n) + q(\varphi,\theta) G(k,n)
        \nonumber
    \end{equation}
    where
    \begin{equation}
    \begin{aligned}
    F(k,n) &= 4 k^2-4 k (n-1)+n^2-3 n+2, \\
    G(k,n) &= k (n-k-1).
    \end{aligned}
    \nonumber
    \end{equation}
    Their explicit expressions can be found in App.~\ref{app:offdiagonals}. After the substitution $k = (n-1)x$, dividing by $(n-1)^2$ and taking the thermodynamic limit, we obtain 

    \begin{equation}
    \begin{aligned}
    \frac{F((n{-}1)x,n)}{(n-1)^2} \xrightarrow{n \to \infty} 4x^2-4x+1 &\eqqcolon F_1(x),\\  \frac{G((n{-}1)x,n)}{(n-1)^2} = -x^2+x &\eqqcolon G_1(x).        
    \end{aligned}
    \nonumber
    \end{equation}
    In the thermodynamic limit, the contributions of the single-body and two-body PI measurement operators vanish, as they are constant and linear in $k,n$ respectively. 
    Thus we only need to consider linear combinations of $F_1(x)$ and $G_1(x)$.
    Let us denote by $\vec{A}_1$ the leading order contribution of $\vec{A}$ and similarly for $\vec{B}_1$ and $\vec{C}_1$, the limiting $x^2$-, $x$- and constant-coefficients of $h_1(x,\vec{\alpha})$.
    The inherited relation then reads $\vec{B}_1\cdot\vec{\alpha} = -\vec{A}_1\cdot\vec{\alpha}$ for all $\vec{\alpha}$, so we have that $\vec{A}_1 = -\vec{B}_1$.
    We note we can express the leading order of $\mathbf{Q}$ as $\mathbf{Q}_{1} = \vec{B_1}\vec{B_1}^{\top}-2\left(\vec{A_1}\vec{C_1}^{\top} + \vec{C_1}\vec{A_1}^{\top}\right) $.
    Now substituting $\vec{B}_1 = - \vec{A}_1$, we obtain $\mathbf{Q}_1 = \vec{A_1}\vec{A_1}^{\top}- 2\left(\vec{A_1}\vec{C_1}^{\top} + \vec{C_1}\vec{A_1}^{\top}\right) $.
    Let us define $\Psi \coloneqq\left( \vec{A}_1 \; \vec{C}_1 \right) \in \mathbbm{R}^{9 \times 2}$, then $\mathbf{Q}_1 = \Psi \mathbf{M} \Psi^T$ where $\mathbf{M} = \bigl( \begin{smallmatrix} 1 & -2 \\ -2 & 0 \end{smallmatrix} \bigr) $.
    As $\det(\mathbf{M}) = -4 < 0$, $\mathbf{M}$ is of rank at most $2$. The maximum rank is achieved when $\vec{A}_{1}$ and $\vec{C}_{1}$ are linearly independent with coincides with  $\theta\not\equiv\varphi\pmod{\pi}$: one checks from App.~\ref{app:ThreeBodyHPlanes} that the $(\alpha_{000},\alpha_{111})$ minor of $\Psi$ equals $\tfrac{9}{2}\left(\cos2\varphi-\cos2\theta\right)\sin\varphi\sin\theta$ and that all six minors vanish simultaneously precisely when $\theta\equiv\varphi\pmod{\pi}$. Therefore, for $\theta\not\equiv\varphi\pmod{\pi}$, $\Psi$ has full column rank, and $\mathbf{Q}_1 = \Psi\mathbf{M}\Psi^{\top}$ inherits the rank and signature of $\mathbf{M}$ by Sylvester's law of inertia, so $\mathbf{Q}_1$ has rank $2$ and signature $(1,1)$.
\end{proof}

The fact that $\mathbf{Q}_1$ is rank $2$ lets the leading-order condition be written in only two variables rather than the nine Bell coefficients. 
As stated in Theorem~\ref{thm:rank2}, $\mathbf{Q}_1$ is indefinite, of signature $(1,1)$.
A rank-two indefinite quadratic form is a product of two linear forms, so the homogeneous quadric $\vec{\alpha}^{\top}\mathbf{Q}_1\vec{\alpha}=0$ is \emph{degenerate}, a pair of hyperplanes. Concretely, with $a=\vec{A}_1\!\cdot\!\vec{\alpha}$ and $c=\vec{C}_1\!\cdot\!\vec{\alpha}$ one finds $\vec{\alpha}^{\top}\mathbf{Q}_1\vec{\alpha}=a^{2} - 4ac=a(a-4c)$, so in the strict thermodynamic limit the interior condition is bounded by the two hyperplanes $a=0$ and $4c=a$, not by a smooth surface. The curvature is recovered only from the full quadratic form of Theorem~\ref{thm:QuadraticSurface}: in the coordinates $a=\vec{A}\!\cdot\!\vec{\alpha}$, $b=\vec{B}\!\cdot\!\vec{\alpha}$, $c=\vec{C}\!\cdot\!\vec{\alpha}$ the interior condition reads $b^{2}\le 4ac$, governed by the non-degenerate ternary form $b^{2}-4ac$ of signature $(2,1)$. The associated quadric is a cone, whose affine plane sections are ellipses, parabolas or hyperbolas according to the slicing direction. Since the leading order enforces $b=-a$ (Theorem~\ref{thm:rank2}), the physical operators lie on a slice on which this cone degenerates to the two hyperplanes above.
The curvature of the three-body boundary is therefore an $O(1/n)$ effect, most pronounced at the moderate system sizes relevant to experiments. 

The two regimes are illustrated in Fig.~\ref{fig:threebody}, drawn in the polynomial-space coordinates $(s,b)=(a-c,\,b)$, in which the picture is the same for every $(\varphi,\theta)$: the circle and the tangent-line collapse follow purely from $\mathcal{h}_1$ being a quadratic nonnegative on $[0,1]$, the settings entering only through the coefficient maps (cf.\ the Remark of Sec.~\ref{sec:sdp:sos}). The eigendecomposition of $\mathbf{Q}_1$, including closed-form eigenvalues and eigenvectors, is worked out in Apps.~\ref{app:Eigvec} and~\ref{app:Q1diag}. 

In many-body experiments, the thermodynamic limit cannot be fully attained. 
Therefore, for the application of our results to these scenarios, we bound the error induced by taking the continuum limit.
We first note that for arbitrary functions, the continuous maximizer $(n-1) x^*$ and the discrete maximizer $k^*$ do not need to be the same when $x^*$ lies in the interior of the interval $[0,1]$. 
However, as $h_1(k,\vec{\alpha})$ is a quadratic function, the supremum of $h_1(k,\vec{\alpha})$ can only be attained in the interior when the function is concave. 
In this case, the discrete maximizer $k^*$ is the closest integer to $(n-1) x^*$ and we can avoid any error by rounding $(n-1) x^*$ appropriately. 
We now focus on the error introduced by the approximation $\mathbf{Q} \approx \mathbf{Q}_1$.
This is given by 
\begin{widetext}
\begin{equation}
\left|\vec{\alpha}^T \mathbf{Q}\,\vec{\alpha}
- \vec{\alpha}^T \mathbf{Q}_1\vec{\alpha}\right|
=
\left|\vec{\alpha}^T\!\left(\tfrac{1}{n-1}\mathbf{Q}_2
+ \tfrac{1}{(n-1)^2}\mathbf{Q}_3\right)\!\vec{\alpha}\right|
\leq
\left(\frac{\|\mathbf{Q}_2\|_2}{n-1}
+ \frac{\|\mathbf{Q}_3\|_2}{(n-1)^2}\right)\|\vec{\alpha}\|^2 \eqqcolon \Delta.
\end{equation}
\end{widetext}
The bound follows from the triangle inequality together with $|\vec{\alpha}^{T}\mathbf{A}\vec{\alpha}| \leq \|\mathbf{A}\|_2\,\|\vec{\alpha}\|^2$.
The error can thus be quantified by the spectral norms $\|\mathbf{Q}_2\|_2$ and $\|\mathbf{Q}_3\|_2$ of Eq.~\eqref{eq:Qclosed}, which depend only on the measurement parameters $(\varphi, \theta)$, in particular $\|\mathbf{Q}_2\|_2=2\bigl(\lVert\vec{A}_1\rVert\lVert\vec{C}_1\rVert+|\vec{A}_1\cdot\vec{C}_1|\bigr)$ in closed form. 
As the stoquasticity cone is homogeneous in $\vec{\alpha}$, we may without loss of generality set $\|\vec{\alpha}\|=1$.
The rate of convergence of the error is thus $O(n^{-1})$.
In practice, for systems of size $n$, the leading-order term $\vec{\alpha}^T\mathbf{Q}_1\vec{\alpha}\leq 0$ is thus a reliable approximation whenever $n \gg \|\mathbf{Q}_2\|_2/\epsilon$, where $\epsilon$ is determined by  $\vec{\alpha}^T\mathbf{Q}_1\vec{\alpha} = 
-\epsilon\|\vec{\alpha}\|^2$.
Consequently, points deeper inside the leading-order cone (larger $\epsilon$) require smaller $n$ for the approximation to be reliable.
As a concrete example, we consider the measurement parameters $(\varphi, \theta) = (\pi/6, 5\pi/6)$ which were used in the Bell experiments involving $480$ ${}^{87}\mathrm{Rb}$ atoms in a Bose-Einstein condensate~\cite{schmied_bell_2016} and $5\cdot 10^5$ atoms in a thermal ensemble~\cite{engelsen_bell_2017}.
The relative deviation $\Delta/\|\mathbf{Q}_1\|_2$ is below $2.8 \cdot 10^{-2}$ at $n=480$ and below $2.7 \cdot 10^{-5}$ at $n = 5 \cdot 10^5$. At the boundary points $x \in \{0,1\}$, the conditions are obtained by substituting $x=0$ and $x=1$ into $\mathcal{h}_{d}(x,\vec{\alpha}) \leq 0$, and no approximation is made.
The endpoint (linear) conditions are therefore exact at every $n$, and the $O(1/n)$ deviation $\Delta$ affects only the interior conic condition.

\section{Conclusion and outlook}

In this work we have shown that the stoquasticity conditions for PI Bell operators, which naively require checking $O(n)$ constraints per off-diagonal, collapse to a single condition per off-diagonal. The constraints along each off-diagonal share an underlying continuous structure, and their supremum is the nonnegativity of a univariate polynomial on an interval, which the Markov--Luk\'acs theorem converts into an \emph{exact} linear matrix inequality, with no relaxation hierarchy.
The resulting feasibility region is a projected spectrahedron at every order $K$, decided by a single semidefinite program that certifies stoquasticity already at finite $n$.
For two-body operators this recovers the known three-hyperplane characterization of Ref.~\cite{li_stoquastic_2026} exactly, at every $n$.
For three-body operators, it reveals a qualitatively new feature: the interior stoquasticity boundary is a quadric, that is, a cone in coefficient space rather than a hyperplane, described by a canonical linear matrix inequality with a single, provably necessary, auxiliary variable (App.~\ref{app:lift}).
Its curvature is an $O(1/n)$ effect, degenerating to a pair of hyperplanes in the strict thermodynamic limit.
We have quantified how accurately the thermodynamic-limit description approximates finite systems, and, for the subregion of Bell coefficients whose stoquasticity conditions remain purely linear, given a solver-free criterion based on Sturm's theorem (App.~\ref{app:boundary}).

Several directions remain open.
The conic geometry of the three-body interior regime
provides a compact description of the stoquasticity boundary, 
but a complete characterization of the cone also requires its 
vertex representation, which encodes the extremal Bell operators.
For the two-body case, this was found in
Ref.~\cite{li_stoquastic_2026}.
Finding the analogous description for the three-body cone at finite $n$ is a natural next step (in the strict limit the cone is polyhedral and its vertex description elementary), and the conic structure derived here directly constrains what the answer can look like.

The subleading corrections to this conic geometry may also be
physically relevant.
Optimizing the quantum-classical gap, which determines how strongly 
a Bell operator can distinguish quantum from classical correlations, 
may depend on corrections beyond the leading-order term.
The closed forms of Eq.~\eqref{eq:Qclosed} make these corrections explicit.
Understanding how the geometry of the stoquasticity boundary
at finite $n$ affects this optimization would connect the results 
here directly to experimentally relevant questions.

Extending the framework to operators of order $K\geq 4$
remains an open problem.
While we have established that the off-diagonal constraints are 
polynomials of degree at most $K-d$ for general $K$, the explicit 
expressions of these polynomials are not yet known beyond $K=3$.
Deriving these expressions would be the natural next step toward 
a systematic geometric characterization of the stoquasticity cone 
at arbitrary operator order, and would allow the interior-regime
analysis developed here for $K=3$ to be carried out for higher-order
operators.
Note that Lemma~\ref{lem:parity} already constrains the outcome: the strict-limit cones remain polyhedral through $K=4$, and the first curved thermodynamic-limit boundary can appear only at $K=5$,  so that for all lower orders the curved geometry resides entirely in the finite-$n$ corrections.

A second axis of generalization is the local dimension.
The Bell correlations now being probed in spin-1 atomic
ensembles~\cite{kobus_multisetting_2025} are governed by
three-outcome PI Bell operators~\cite{aloy_deriving_2024, 
muller-rigat_three-outcome_2024, 
aloy_nonlocality_2026},
and numerical evidence indicates that stoquasticity remains
the relevant structural property there: the qutrit state of
maximal violation stays concentrated and smooth on the
symmetric Dicke basis, consistent with a sign-free,
nonnegative-amplitude optimum (see Fig.~3 of
Ref.~\cite{muller-rigat_three-outcome_2024}.
The geometry, however, changes qualitatively.
For qubits the symmetric sector is a one-dimensional chain,
each off-diagonal is a univariate polynomial, and its
nonnegativity on an interval is a linear
matrix inequality, the single-SDP exactness exploited
throughout this work.
For $n$ qutrits the symmetric sector is instead the
two-dimensional simplex of partitions $(\mu_0,\mu_1,\mu_2)$
of $n$, and stoquasticity becomes the
nonnegativity of a \emph{multivariate} polynomial on that
simplex.
In two or more variables nonnegative polynomials are no
longer sums of squares in general, so the exact representability is lost and
one must instead climb a convergent Positivstellensatz
hierarchy,
exact only in the limit.
Characterizing the qutrit stoquasticity cone is therefore a
natural and experimentally motivated, but harder
problem. 

\section*{Code availability}

The supporting Mathematica notebooks used in this project can be found at: \url{https://github.com/janli11/ExactSDPStoqPIBell}

\section*{AI statement}

The authors used generative AI tools, including ChatGPT and Claude, during the preparation of this manuscript. These tools were used as aids for tasks such as exploring algebraic manipulations and suitable formulations of intermediate bounds, suggesting approaches for symbolic calculations and Mathematica implementations, checking or developing intermediate proof arguments, and assisting with the organization and presentation of the manuscript. The scientific questions, principal results, proofs, and conclusions were developed by the authors. All AI-assisted calculations, arguments, and text incorporated into the manuscript were independently checked and verified by the authors, who take full responsibility for the content of the work.

\acknowledgments
J.L., E.v.N., and J.T. acknowledge the support received by the Dutch National Growth Fund
(NGF), as part of the QDNL programme. 
J.T. acknowledges the support received from the European Union's Horizon Europe research and innovation programme through the ERC StG FINE-TEA-SQUAD (Grant No.~101040729). 
This publication is part of the `Quantum Inspire - the Dutch Quantum Computer in the Cloud' project (with project number [NWA.1292.19.194]) of the NWA research program `Research on Routes by Consortia (ORC)', which is funded by the Netherlands Organisation for Scientific Research (NWO).
This work was supported by the Netherlands Organization for Scientific Research (NWO/OCW), as part of Quantum Limits (project number SUMMIT.1.1016).

The views and opinions expressed here are solely
those of the authors and do not necessarily reflect those of the funding institutions. 
None of the funding institutions can be held responsible for them.

\clearpage
\pagestyle{plain}
\bibliographystyle{apsrev4-1}
\bibliography{referencesZotero}

\newpage
\appendix
\onecolumngrid

\section{Matrix elements of PI measurement operators}\label{app:offdiagonals}

Throughout the appendices we work on a general spin-$J$ block and write $\Gamma(k) \coloneqq \sqrt{(k+1)(2J-k)}$ and $\Xi(k) \coloneqq 2J-2k$.
On the symmetric block, $2J=n$, these are the functions of Eq.~\eqref{eq:singlePauliel}. Let us start by reminding that a PI Bell operator 
\begin{equation}
\mathbf{B} \coloneqq \sum_{\vec{x}\,\in\,\mathcal{X}_K} \alpha_{\vec{x}}\,
\mathbf{S}_{\vec{x}},
\qquad
\mathcal{X}_K \coloneqq \bigcup_{k=1}^K \{0,1\}^k,
\end{equation}
is a linear combination of PI measurement operators. 
The explicit expression for single-body, two-body and three-body PI measurement operators are
\begin{equation}
    \label{eq:SingleTwoThreeBodyPIMeasOP}
     \mathbf{S}_{x_1} := \sum_{i=0}^{n-1} \mathbf{M}_{x_1}^{i},
     \; \quad \;
     \mathbf{S}_{x_1 x_2}:=\sum_{\substack{i_1,i_2=0\\i_1\neq i_2}}^{n-1} \mathbf{M}_{x_1}^{i_1} \otimes \mathbf{M}_{x_2}^{i_2},
     \; \quad \;
     \mathbf{S}_{x_1 x_2 x_3}:=\sum_{\substack{i_1,i_2,i_3=0\\ \text{all distinct} }}^{n-1} \mathbf{M}_{x_1}^{i_1} \otimes \mathbf{M}_{x_2}^{i_2} \otimes \mathbf{M}_{x_3}^{i_3},
\end{equation}
where
\begin{equation}\label{eq:qubitmeasop}
\mathbf{M}^i_0=\cos(\varphi) Z^{i}+ \sin(\varphi) X^{i}, \qquad \mathbf{M}^i_1=\cos(\theta) Z^{i}+ \sin(\theta) X^{i}.
\end{equation}
Here $X,Y,Z$ denote the Pauli X, Pauli Y and Pauli Z matrices. 
Let us define the $K$-body PI Pauli operator as 
\begin{align}
   \mathbf{S}_{P_1 P_2 \cdots P_K} \coloneq \sum_{\substack{i_1,...,i_K=0\\ \text{all distinct} }}^{n-1} P_1^{i_1} \otimes P_2^{i_2} \otimes \cdots \otimes P_K^{i_K},
   \label{eq:Paulicorrelator}
\end{align}
where $P_j \in \{X,Y,Z\}$.
We call $K$ the degree of the $K$-body PI Pauli operator.
An important fact is that any $K$-body PI measurement operators, and thus also $K$-body PI Bell operators, can be written as a linear combination of PI Pauli operators of degree at most $K$.
To see this above, let us start by noting that the single-body permutationally invariant measurement operators can be written as 
\begin{align}
\mathbf{S}_0 = \cos\varphi \, \mathbf{S}_Z + \sin\varphi \, \mathbf{S}_X, \quad
\mathbf{S}_1 = \cos\theta \, \mathbf{S}_Z + \sin\theta \, \mathbf{S}_X, \nonumber
\end{align}
by plugging Eq.~\ref{eq:qubitmeasop} into the first term of Eq.\eqref{eq:SingleTwoThreeBodyPIMeasOP}, where
\begin{align}
\mathbf{S}_Z = \sum_{i=0}^{n-1} Z^i, \quad
\mathbf{S}_X = \sum_{i=0}^{n-1} X^i. \nonumber
\end{align}
Doing the same for two-body permutationally invariant measurement operators of Eq.~\ref{eq:SingleTwoThreeBodyPIMeasOP}, we obtain
\begin{align}
\mathbf{S}_{00} &= \cos^2\varphi \, \mathbf{S}_{ZZ} + 2 \cos\varphi \sin\varphi \, \mathbf{S}_{ZX} + \sin^2\varphi \, \mathbf{S}_{XX}, \nonumber\\
\mathbf{S}_{01} &= \cos\varphi \cos\theta \, \mathbf{S}_{ZZ} + (\cos\varphi \sin\theta + \sin\varphi \cos\theta) \mathbf{S}_{ZX} + \sin\varphi \sin\theta \, \mathbf{S}_{XX}, \nonumber\\
\mathbf{S}_{11} &= \cos^2\theta \, \mathbf{S}_{ZZ} + 2 \cos\theta \sin\theta \, \mathbf{S}_{ZX} + \sin^2\theta \, \mathbf{S}_{XX}, \nonumber
\end{align}
and similarly for the three-body permutationally invariant measurement operators of Eq.~\ref{eq:SingleTwoThreeBodyPIMeasOP}  we obtain
\begin{align}
\mathbf{S}_{000} &= \cos^3\varphi \, \mathbf{S}_{ZZZ} + 3 \cos^2\varphi \sin\varphi \, \mathbf{S}_{ZZX} + 3 \cos\varphi \sin^2\varphi \, \mathbf{S}_{ZXX} + \sin^3\varphi \, \mathbf{S}_{XXX}, 
\nonumber\\
\mathbf{S}_{001} &= \cos^2\varphi \cos\theta \, \mathbf{S}_{ZZZ} + \left( 2 \cos\theta \cos\varphi \sin\varphi + \cos^2\varphi \sin\theta \right) \mathbf{S}_{ZZX} \nonumber\\
&\quad+ \left( 2 \sin\theta \sin\varphi \cos\varphi + \sin^2\varphi \cos\theta \right) \mathbf{S}_{ZXX} + \sin^2\varphi \sin\theta \, \mathbf{S}_{XXX},
\nonumber\\
\mathbf{S}_{011} &= \cos^2\theta \cos\varphi \, \mathbf{S}_{ZZZ} + \left( 2 \cos\varphi \cos\theta \sin\theta + \cos^2\theta \sin\varphi \right) \mathbf{S}_{ZZX} \nonumber\\
&\quad+ \left( 2 \sin\varphi \sin\theta \cos\theta + \sin^2\theta \cos\varphi \right) \mathbf{S}_{ZXX} + \sin^2\theta \sin\varphi \, \mathbf{S}_{XXX},
\nonumber\\
\mathbf{S}_{111} &= \cos^3\theta \, \mathbf{S}_{ZZZ} + 3 \cos^2\theta \sin\theta \, \mathbf{S}_{ZZX} + 3 \cos\theta \sin^2\theta \, \mathbf{S}_{ZXX} + \sin^3\theta \, \mathbf{S}_{XXX}. \nonumber
\end{align}
This extends similarly to higher order PI measurement operators, meaning that they can be expressed as a linear combination of PI Pauli operators with weights given by trigonometric functions of $\varphi, \theta$. 
We note that for the PI Pauli operators defined in Eq. \eqref{eq:Paulicorrelator}, only the number of distinct Pauli operators $X,Y,Z$ is relevant and not the order.
For example $\mathbf{S}_{ZXX} = \mathbf{S}_{XZX} = \mathbf{S}_{XXZ}$.

We note that the PI measurement operators $\mathbf{S}_{\vec{x}}$, and therefore also Bell operators $\mathbf{B}$, are symmetric matrices. 
This is because the Pauli matrices $\mathbf{X},\mathbf{Z}$ are symmetric in the computational basis and therefore the local measurement operators $\mathbf{M}^i_x$ are as well.
Symmetry is preserved under taking tensor products and linear combinations, thus PI measurement operators $\mathbf{S}_{\vec{x}}$ and consequently the Bell operators $\mathbf{B}$ are symmetric operators in the computational basis. 
Additionally, in the block diagonal form $\mathbf{B} = \oplus_J B_J$, cf. Sec.~\ref{sec:PIBell} in the main text, the basis of each subblock, given by Eq.~\eqref{eq:DickeBasis}.
Since this basis is related to the computational basis by a real change of basis, the block diagonal form is hence also symmetric.

\section{Higher-body to single-body PI Pauli operator identities}\label{app:hightosingle}

To calculate the matrix elements of the higher order PI Pauli operators, it is convenient to write them in terms of products of single-body PI Pauli operators.
The decompositions of two-body and three-body PI Pauli operators in terms of single-body Pauli operators are given by~\cite{li_stoquastic_2026}

\begin{align}
    \mathbf{S}_{ZZ} = \mathbf{S}_Z \mathbf{S}_Z - n \mathbbm{I}, \quad
    \mathbf{S}_{ZX} = \mathbf{S}_Z \mathbf{S}_X - i\,\mathbf{S}_Y, \quad
    \mathbf{S}_{XX} = \mathbf{S}_X \mathbf{S}_X - n \mathbbm{I} \nonumber
\end{align}
for the two-body PI Pauli operators and
\begin{align}
    \mathbf{S}_{ZZZ} &= (\mathbf{S}_Z)^3 + (2-3n)\mathbf{S}_Z ,\nonumber\\
    \mathbf{S}_{ZZX} &= (\mathbf{S}_Z)^2 \mathbf{S}_X - n \mathbf{S}_X- i \mathbf{S}_{Y} \mathbf{S}_Z - i \mathbf{S}_{Z} \mathbf{S}_{Y} ,\nonumber\\
    \mathbf{S}_{ZXX} &= \mathbf{S}_Z (\mathbf{S}_X)^2 +(2-n) \mathbf{S}_Z - 2i \mathbf{S}_Y \mathbf{S}_X ,\nonumber\\
    \mathbf{S}_{XXX} &= (\mathbf{S}_X)^3 + (2-3n)\mathbf{S}_X \nonumber
\end{align}
for the three-body PI Pauli operators.

\subsection{Higher order operators}\label{app:sub:hightosingleprod}

We will now show that each $K$-body PI Pauli operator, can be decomposed into products of single-body PI Pauli operators. 
To see this, we note that we can always do the following
\begin{equation}
\begin{aligned}
    \mathbf{S}_{P_1 P_2 \ldots P_K} = &\sum_{|\{i_1, i_2,\ldots,i_K\}|=K} P_1^{i_1}  P_2^{i_2}  \ldots P_K^{i_K} \nonumber \\
    = &\sum_{|\{i_1, i_2,\ldots,i_K\}|=K} P_1^{i_1}  P_2^{i_2}  \ldots P_K^{i_K}
    + \sum_{\substack{|\{i_1, i_2,\ldots,i_{K-1}\}|=K-1 \\
                i_K \in \{i_1, i_2,\ldots,i_{K-1}\}}}
                P_1^{i_1}  P_2^{i_2}  \ldots P_K^{i_K} - \sum_{\substack{|\{i_1, i_2,\ldots,i_{K-1}\}|=K-1 \\ 
                i_K \in \{i_1, i_2,\ldots,i_{K-1}\}}}
                P_1^{i_1}  P_2^{i_2}  \ldots P_K^{i_K} \nonumber\\
    = &\sum_{|\{i_1, i_2,\ldots,i_{K-1}\}|=K-1,i_K} P_1^{i_1}  P_2^{i_2}  \ldots P_K^{i_K} 
    - \sum_{\substack{|\{i_1, i_2,\ldots,i_{K-1}\}|=K-1 \\ 
                i_K \in \{i_1, i_2,\ldots,i_{K-1}\}}}
                P_1^{i_1}  P_2^{i_2}  \ldots P_K^{i_K} \nonumber\\
    = &\sum_{|\{i_1, i_2,\ldots,i_{K-1}\}|=K-1} P_1^{i_1}  P_2^{i_2}  \ldots P_{K-1}^{i_{K-1}} \sum_{i_K} P_K^{i_K} \nonumber\\
    &\quad - \sum_{\substack{|\{i_1, i_2,\ldots,i_{K-1}\}|=K-1 \\ 
                i_K \in \{i_1, i_2,\ldots,i_{K-1}\}}}
                P_1^{i_1}  P_2^{i_2}  \ldots P_K^{i_K}.
\end{aligned}
\end{equation}
The second term in the last line is a $(K-1)$-th order PI Pauli operator as $i_K \in \{i_1, i_2,\ldots,i_{K-1}\}$.
The term $\sum_{|\{i_1, i_2,\ldots,i_{K-1}\}|=K-1} P_1^{i_1}  P_2^{i_2}  \ldots P_{K-1}^{i_{K-1}} \sum_{i_K} P_K^{i_K}$ is a product of a single-body PI Pauli operator and a $(K-1)$-th order PI Pauli operator.
This process can be applied iteratively to the resulting lower-order operators until we are left only with products of single-body PI Pauli operators. 

\section{Matrix elements of Two- and Three-body PI Pauli operators}

The explicit expressions of the off-diagonal elements of the three-body PI measurements operators are given below~\cite{li_stoquastic_2026}.
As the operators are symmetric, if is sufficient to only give one elements on one side of the diagonal $(k,k+d)$ with $d >0$.
The functions characterizing the matrix elements of the first off-diagonal are given by
\begin{align}
\mathbf{S}_{000}^{k,k+1} = &3 \Gamma(k) \sin (\phi ) \left(\left(4 k^2-4 k (n-1)+n^2-3 n+2\right) \cos ^2(\phi )+k (-k+n-1) \sin ^2(\phi )\right),\nonumber\\ 
\mathbf{S}_{001}^{k,k+1} =& \Gamma(k) \bigg[\left(4 k^2-4 k (n-1)+n^2-3 n+2\right) \cos (\phi ) \left(\sin (\theta ) \cos (\phi )+2 \cos (\theta ) \sin (\phi ) \right) \nonumber \\
& - 3 k \sin (\theta ) (k-n+1) \sin ^2(\phi )\bigg],\nonumber\\ 
\mathbf{S}_{011}^{k,k+1} =& \Gamma(k) \bigg[\cos (\theta ) \left(4 k^2-4 k (n-1)+n^2-3 n+2\right) \left(2 \sin (\theta ) \cos (\phi )+\cos (\theta ) \sin (\phi )\right) \nonumber \\
&-3 k \sin ^2(\theta ) (k-n+1) \sin (\phi )\bigg],\nonumber\\ 
\mathbf{S}_{111}^{k,k+1} =& 3 \sin (\theta ) \Gamma(k) \left(\cos ^2(\theta ) \left(4 k^2-4 k (n-1)+n^2-3 n+2\right)+k \sin ^2(\theta ) (-k+n-1)\right),\nonumber
\end{align}
the functions characterizing the matrix elements of the second off-diagonals are given by
\begin{align}
\mathbf{S}_{000}^{k,k+2} &= 3 \Gamma(k+1)  \Gamma(k)  (-2 k+n-2) \sin ^2(\phi ) \cos (\phi ),\nonumber\\ 
\mathbf{S}_{001}^{k,k+2} &= \Gamma(k+1)  \Gamma(k) (-2 k+n-2) \sin (\phi ) (2 \sin (\theta ) \cos (\phi )+\cos (\theta ) \sin (\phi )),\nonumber\\ 
\mathbf{S}_{011}^{k,k+2} &= \sin (\theta ) \Gamma(k+1)  \Gamma(k) (-2 k+n-2) (\sin (\theta ) \cos (\phi )+2 \cos (\theta ) \sin (\phi )),\nonumber\\ 
\mathbf{S}_{111}^{k,k+2} &= 3 \sin ^2(\theta ) \cos (\theta ) \Gamma(k+1)  \Gamma(k) (-2 k+n-2),\nonumber
\end{align}
and functions characterizing the elements of the third off-diagonals are given by
\begin{align}
\mathbf{S}_{000}^{k,k+3} &= \Gamma(k+2) \Gamma(k+1)  \Gamma(k) \sin ^3(\phi ),\nonumber\\ 
\mathbf{S}_{001}^{k,k+3} &= \sin (\theta ) \Gamma(k+2) \Gamma(k+1)  \Gamma(k) \sin ^2(\phi ),\nonumber\\ 
\mathbf{S}_{011}^{k,k+3} &= \sin ^2(\theta ) \Gamma(k+2) \Gamma(k+1)  \Gamma(k) \sin (\phi ),\nonumber\\ 
\mathbf{S}_{111}^{k,k+3} &= \sin ^3(\theta ) \Gamma(k+2) \Gamma(k+1)  \Gamma(k).\nonumber
\end{align}
We also give the matrix elements of the first off-diagonal of the single-body and two-body PI measurement operators
\begin{align}
\mathbf{S}_0^{k,k+1} &= \Gamma(k) \sin (\phi),\nonumber\\ 
\mathbf{S}_1^{k,k+1} &= \Gamma(k) \sin (\theta), \nonumber\\ 
\mathbf{S}_{00}^{k,k+1} &=\Gamma(k) (-2 k+n-1) \sin (2 \phi ), \nonumber\\ 
\mathbf{S}_{01}^{k,k+1} &= \Gamma(k) (-2 k+n-1) \sin ( \phi + \theta ), \nonumber\\ 
\mathbf{S}_{11}^{k,k+1} &= \Gamma(k) (-2 k+n-1) \sin (2 \theta ). \nonumber
\end{align}

\section{Q-decomposition}\label{app:Qdecomp}

Let us define
\begin{align}
    \vec{e}_1 &\coloneq (\sin (\phi), \sin (\theta), 0,0,0,0,0,0,0), \nonumber \\
    \vec{e}_2 &\coloneq (0,0,\sin (2 \phi ), \sin ( \phi + \theta ), \sin (2 \theta ),0,0,0,0), \nonumber \\
    \vec{e}_{3,F} &\coloneq (0,0,0,0,0,
    3 \sin (\phi ) \cos ^2(\phi ),
    \cos (\phi ) (\sin (\theta ) \cos (\phi )+2 \cos (\theta ) \sin (\phi )),\nonumber\\
    & \hspace{2cm} \cos (\theta ) (2 \sin (\theta ) \cos (\phi )+\cos (\theta ) \sin (\phi ))
    ,3 \sin (\theta ) \cos ^2(\theta )), \nonumber\\
    \vec{e}_{3,G} &\coloneq (0,0,0,0,0,\sin ^3(\phi ),\sin (\theta ) \sin ^2(\phi ),\sin ^2(\theta ) \sin (\phi ),\sin ^3(\theta ),
\end{align}
which characterize the trigonometric coefficients of the first off-diagonal of the PI measurement operators.
Additionally, let us define 
\begin{align}
    F(k,n) \coloneq 4 k^2-4 k (n-1)+n^2-3 n+2, \quad \text{and} \quad
    G(k,n) \coloneq k (-k+n-1).
\end{align}
The polynomial characterizing the stoquasticity condition of the first off-diagonal can thus be written as 
\begin{equation}
   h_{1}(n,k;\varphi,\theta,\vec{\alpha}) = \big( \vec{e}_1 + (-2 k+n-1) \vec{e}_2  + F(k,n) \vec{e}_{3,F} + G(k,n) \vec{e}_{3,G}\big) \cdot \vec{\alpha}.
\end{equation}
Making the substitution $k\to (n-1)x$, and dividing by $(n-1)^{2}$, we obtain 
\begin{equation}
   \mathcal{h}_{1}(x,\vec{\alpha}) = \big( \frac{1}{(n-1)^2} \vec{e}_1 + \frac{1}{n-1} (1-2x) \vec{e}_2  + \mathcal{F}(x,n) \vec{e}_{3,F} + \mathcal{G}(x,n) \vec{e}_{3,G}\big) \cdot \vec{\alpha},
\end{equation}
where
\begin{equation}
    \mathcal{F}(x,n) \coloneq 4 x^2-4 x + 1 -\frac{1}{(n-1)} \quad \text{and} \quad \mathcal{G}(x,n) \coloneq x(1-x).
\end{equation}
Grouping terms such it takes the form $ \mathcal{h}_{1}(x,\vec{\alpha}) = \vec{A} \cdot \vec{\alpha} \; x^2 + \vec{B} \cdot \vec{\alpha} \; x  + \vec{C} \cdot \vec{\alpha},$
we obtain the following expressions 
\begin{align}
    \vec{A} &= 4\vec{e}_{3,F} - \vec{e}_{3,G}, \nonumber\\
    \vec{B} &= - \vec{A} - \frac{2}{n-1} \vec{e}_2, \nonumber \\
    \vec{C} &= (1 -\frac{1}{(n-1)}) \vec{e}_{3,F} + \frac{1}{n-1} \vec{e}_2 + \frac{1}{(n-1)^2} \vec{e}_1.
\end{align}
Let us denote by $\vec{C}_1$ the leading order term of $\vec{C}$ in $\frac{1}{n-1}$ and similarly for $\vec{A}$ and $\vec{B}$, i.e. $\vec{C}_1 = \vec{e}_{3,F}$, $\vec{A}_1 = \vec{A}$ and $\vec{B}_1 = - \vec{A}$.
Plugging the expression for $\vec{A}, \vec{B}, \vec{C}$ into Eq.~\eqref{eq:QForm} of the main text, we obtain the following decomposition
\begin{equation}
\mathbf{Q}
=
\mathbf{Q}_1
+
\frac{1}{n-1}\mathbf{Q}_2
+
\frac{1}{(n-1)^2}\mathbf{Q}_3,
\end{equation}
with
\begin{align}
\mathbf{Q}_1 &= \vec{A}_1\vec{A}_1^{\top} - 2\left(\vec{A}_1\vec{C}_1^{\top} + \vec{C}_1\vec{A}_1^{\top} \right), \nonumber\\
\mathbf{Q}_2 &= 2\left( \vec{A}_1\vec{C}_1^{\top} + \vec{C}_1\vec{A}_1^{\top} \right), \nonumber\\
\mathbf{Q}_3 &= 4\vec{e}_2\vec{e}_2^{\top} - 2\left( \vec{A}_1\vec{e}_1^{\top} + \vec{e}_1\vec{A}_1^{\top} \right).
\end{align}

\section{Coefficient vectors of the three-body off-diagonal}\label{app:ThreeBodyHPlanes}

The first off-diagonal of the three-body PI Bell operator takes the form $\mathcal{h}_1(x,\vec{\alpha}) = x^{2}\,\vec{A}\cdot\vec{\alpha} + x\,\vec{B}\cdot\vec{\alpha} + \vec{C}\cdot\vec{\alpha}$ of Eq.~\eqref{eq:ThreeBodyFirstDiag}. After the substitution $k\to (n-1)x$, division by $(n-1)^{2}$, and the limit $n\to\infty$, the single- and two-body contributions vanish, being respectively constant and linear in $k,n$, so that only the three-body coefficients $(\alpha_{000},\alpha_{001},\alpha_{011},\alpha_{111})$ survive. The nonvanishing components of $\vec{A}$ and $\vec{C}$ are
\begin{align}
A_{000} &= \tfrac{3}{2}\left(3+5\cos 2\varphi\right)\sin\varphi, & C_{000} &= 3\sin\varphi\cos^{2}\varphi, \nonumber\\
A_{001} &= \tfrac{1}{2}\sin\theta - \tfrac{1}{4}\sin(\theta-2\varphi) + \tfrac{15}{4}\sin(\theta+2\varphi), & C_{001} &= \tfrac{1}{2}\cos\varphi\left[3\sin(\theta+\varphi)-\sin(\theta-\varphi)\right], \nonumber\\
A_{011} &= \tfrac{1}{2}\sin\varphi + \tfrac{1}{4}\sin(2\theta-\varphi) + \tfrac{15}{4}\sin(2\theta+\varphi), & C_{011} &= \tfrac{1}{2}\cos\theta\left[3\sin(\theta+\varphi)+\sin(\theta-\varphi)\right], \nonumber\\
A_{111} &= \tfrac{3}{2}\left(3+5\cos 2\theta\right)\sin\theta, & C_{111} &= 3\sin\theta\cos^{2}\theta,
\end{align}
and $\vec{B}=-\vec{A}$ at this order, in agreement with Theorem~\ref{thm:rank2}.

\section{Eigendecomposition of \texorpdfstring{$\mathbf{Q}_1$}{Q1}}\label{app:Eigvec}

The eigenvalues and eigenvectors of $\mathbf{Q}_1$ can be obtained in closed form as follows.
Let $\vec{q}_1, \vec{q}_2$ denote the columns of $\mathbf{Q}_1$ associated with the coefficients $\alpha_{000}$ and $\alpha_{011}$, let $\Phi = (\vec{q}_1\ \vec{q}_2)$ be the $9\times 2$ matrix formed by these columns, and let $V = \operatorname{span}\{\vec{q}_1, \vec{q}_2\}$. Because $\mathbf{Q}_1=\Psi\mathbf{M}\Psi^{\top}$ is symmetric of rank $2$ for $\theta\not\equiv\varphi\pmod{\pi}$(Theorem~\ref{thm:rank2}), its eigenvectors with nonzero eigenvalue span its range. This range is two-dimensional and, since $\vec{q}_1$ and $\vec{q}_2$ are linearly independent (which holds at the settings considered here), it equals $V$. Every nonzero-eigenvalue eigenvector therefore lies in $V$, i.e.\ there exist $\mathbf{u}\in\mathbbm{R}^2$ and $\lambda\neq0$ such that $\mathbf{Q}_1(\Phi\mathbf{u}) = \lambda\,\Phi\mathbf{u}$. Left-multiplying both sides by $\Phi^\top$ gives
\begin{equation}
    \Phi^\top \mathbf{Q}_1\Phi\,\mathbf{u} = \lambda\,\Phi^\top\Phi\,\mathbf{u}.
\end{equation}
Since $\vec{q}_1$ and $\vec{q}_2$ are linearly independent, the Gram matrix $\mathcal{G} = \Phi^\top\Phi$ is invertible, and this reduces to the standard $2\times 2$ eigenvalue problem $\mathbf{N}\mathbf{u} = \lambda\mathbf{u}$, where $\mathbf{N} = \mathcal{G}^{-1}\Phi^\top \mathbf{Q}_1 \Phi$. The eigenvalues are obtained in closed form by solving the characteristic polynomial $\lambda^2 - \tau \lambda + \delta$, where $\tau = \operatorname{Tr}(\mathbf{N})$ and $\delta = \det(\mathbf{N})$, which has solutions
\begin{equation}
    \lambda_\pm = \frac{\tau \pm \sqrt{\tau^2 - 4\delta}}{2}.
\end{equation}
For each eigenvalue, the corresponding $\mathbf{u}_\pm$ is obtained from $\ker(\mathbf{N} - \lambda_\pm\mathbbm{I})$, and the eigenvector of $\mathbf{Q}_1$ in the original space is recovered as $\mathbf{v}_\pm = \Phi\,\mathbf{u}_\pm$.
The explicit coordinates of $\mathbf{u}_\pm$ are recorded in App.~\ref{app:Q1diag}.

\section{Explicit eigenvector coordinates of \texorpdfstring{$\mathbf{Q}_1$}{Q1}}\label{app:Q1diag}

In the notation of App.~\ref{app:Eigvec}, we record the closed-form coordinates of the eigenvectors $\mathbf{v}_\pm=u_{1}^{\pm}\,\vec{q}_1+u_{2}^{\pm}\,\vec{q}_2$ of $\mathbf{Q}_1$. We introduce the following auxiliary quantities
\begin{align*}
\mathcal{F} = &6 \sqrt{2} (3 \sin (\theta ) \cos (\phi )+\cos (\theta ) \sin (\phi )) \sqrt{\mathcal{F}_1}
\end{align*}
where
\begin{align}
    \mathcal{F}_1 = &(\cos (2 \theta )+\cos (2 \phi )-2) \nonumber\\
    &(-4 \cos (2 \theta )+\cos (4 \theta )-4 \cos (2 \phi )+\cos (4 \phi )+6) \nonumber\\
    &(2 \cos (4 \theta ) (113 \cos (2 \phi )+60)+\cos (2 \theta ) (112 \cos (2 \phi )+226 \cos (4 \phi )-163) -64 \sin (2 \theta ) \sin (2 \phi ) \nonumber\\
    &(7 \cos (2 \theta )+7 \cos (2 \phi )+2)+225 \cos (6 \theta )-163 \cos (2 \phi )+120 \cos (4 \phi )+225 \cos (6 \phi )-928)
\end{align}
and
\begin{align}
\mathcal{R} = &3 \Bigl(\sin (\theta ) (-298 \cos (2 \theta )+310 \cos (4 \theta )-90 \cos (6 \theta )+141) 
\cos (\phi )\nonumber\\ 
&+\sin (\theta ) (98 \cos (2 \theta )-58 \cos (4 \theta )-67) 
\cos (3 \phi )\nonumber\\
&+45 \sin (\theta ) \cos (7 \phi )+(\sin (3 \theta )-84 \sin (\theta )) 
\cos (5 \phi )\nonumber\\ &-(50 \cos (\theta )+148 \cos (3 \theta )-228 \cos (5 \theta )+75 \cos (7 \theta )) 
\sin (\phi )\nonumber\\ &+(40 \cos (\theta )+66 \cos (3 \theta )-31 \cos (5 \theta ))
\sin (3 \phi )\nonumber\\ 
&-(56 \cos (\theta )+\cos (3 \theta )) \sin (5 \phi ) \nonumber\\
 &+15 \cos (\theta ) \sin (7 \phi )-1536 \sin ^6(\theta ) \cos (\theta ) \csc (\phi )\Bigr)
\end{align}
We then have that $u_1^+ = \frac{\mathcal{R} + \mathcal{F}}{2} $ and $u_1^- = \frac{\mathcal{R} - \mathcal{F}}{2}$.
\begin{align}
u_2^+ = u_2^- =& 3 (191 \sin (\theta -\phi )-43 \sin (3 (\theta -\phi )) \nonumber\\
&-96 \sin (3 \theta -\phi )+45 \sin (5 \theta -\phi )+387 \sin (\theta +\phi ) \nonumber\\
&-149 \sin (3 (\theta +\phi ))-109 \sin (3 \theta +\phi )-18 \sin (5 \theta +\phi ) \nonumber\\
&+133 \sin (\theta +3 \phi )+45 \sin (5 \theta +3 \phi )-78 \sin (\theta +5 \phi ) \nonumber\\
&+30 \sin (3 \theta +5 \phi )+15 \sin (\theta +7 \phi )+104 \sin (\theta -3 \phi )-39 \sin (\theta -5 \phi ) \nonumber\\
&+7 \sin (3 \theta -5 \phi )+7 \sin (\theta -7 \phi ))
\end{align}

\section{Coefficient matching for the semidefinite certificate}\label{app:sdpgram}

This appendix records the linear coefficient-matching map underlying Lemma~\ref{thm:sos} and carries out the three-body ($K-d=2$) case explicitly, recovering the quadratic form of Theorem~\ref{thm:QuadraticSurface}. In this case, the explicit form of Lemma \ref{thm:sos} reads
\begin{equation}
p_{2}x^{2}+p_{1}x+p_{0}= \begin{pmatrix}
1 & x
\end{pmatrix}\begin{pmatrix}
 g_{00}&g_{01}\\
 g_{01}&g_{11}
\end{pmatrix}
\begin{pmatrix}
1\\
x
\end{pmatrix}
+x(1-x)g,
\end{equation}
with the constraints $G_{0}=\left(\begin{smallmatrix}g_{00}&g_{01}\\ g_{01}&g_{11}\end{smallmatrix}\right)\succeq0,$ $g\geq 0$. We have additional constraints on the coefficients $g$ from the fact that, for a given $i$, the prefactors of $x^{i}$ have to match on the left and right sides of the equation
\begin{equation}
g_{00}=p_0,\qquad 2g_{01}+g=p_1,\qquad g_{11}-g=p_2 .
\end{equation}
Then, from the positive semi-definiteness of $G_{0}$, we get $g_{00}=p_0\ge0$, $g_{11}=p_2+g\geq 0$ and $\det G_0=p_0(p_2+g)-\frac{1}{4}(p_1-g)^{2}\geq 0$. Notice that the existence of $g$ such that $\det G_0\geq 0$ enforces the non-negative discriminant $p_{0}(p_{0}+p_{1}+p_{2})\geq 0$. Since $p_{0}\geq 0$ we get $p_{0}+p_{1}+p_{2}\geq 0$. These two inequalities recover the end-point inequalities $p(0),p(1)\geq 0$.

This, however, is not enough to ensure that the feasibility of the decomposition into Gram matrices is satisfied. If both of the roots of the polynomial $\det G_0=p_0(p_2+g)-\frac{1}{4}(p_1-g)^{2}$ are strictly less than $0$ this clashes with the condition $g\geq 0$. Therefore, the largest root of the polynomial has to be non-negative, a sufficient condition for which reads $p_{1}^{2}-4p_{0}p_{2}\leq 0$. Notice that this is the same condition that we get if we demand that $p_{2}x^{2}+p_{1}x+p_{0}$ has no unique roots, which is the only way to satisfy the constraints arising from Theorem \ref{thm:QuadraticSurface}.

In the notation of Eq.~\eqref{eq:ThreeBodyFirstDiag}, $p_2=-\vec{A}\cdot\vec{\alpha}$, $p_1=-\vec{B}\cdot\vec{\alpha}$ and $p_0=-\vec{C}\cdot\vec{\alpha}$, which gives us
\begin{equation}
p_1^{2}-4p_0p_2\le0
\quad\Longleftrightarrow\quad
(\vec{B}\cdot\vec{\alpha})^{2}-4(\vec{A}\cdot\vec{\alpha})(\vec{C}\cdot\vec{\alpha})\le0,
\end{equation}
which is $\vec{\alpha}^{\top}\mathbf{Q}\vec{\alpha}\le0$ of Theorem~\ref{thm:QuadraticSurface}.
The projection onto the Bell coefficients of the lifted spectrahedron defined by $G_0\succeq0,\,g\ge0$ thus has this three-body conic as its boundary. 

\section{Canonical form for \texorpdfstring{$K=3$}{K=3} and necessity of the lifted variable}\label{app:lift}

For $K=3$ the description of Sec.~\ref{sec:sdp} assembles into a single explicit linear matrix inequality containing one auxiliary scalar, and this scalar cannot be eliminated: the three-body cone is a spectrahedral shadow but provably not a  \emph{spectrahedron}, i.e., it is not a set of the form $\{u\in\mathbbm{R}^{N}:\mathcal{A}(u)\succeq0\}$ with $\mathcal{A}$ an affine map into symmetric matrices.

\subsection{Canonical form}\label{app:lift:canonical}

In the notation of  Eq.~\eqref{eq:ThreeBodyFirstDiag}, write $p_0=-\vec{C}\cdot\vec{\alpha}$, $p_1=-\vec{B}\cdot\vec{\alpha}$, $p_2=-\vec{A}\cdot\vec{\alpha}$ for the coefficients of $p=-\mathcal{h}_1$. Combining the $d=1$ certificate of Eq.~\eqref{eq:sosrep} with the linear conditions of the second and third off-diagonals, membership of $\vec{\alpha}$ in the three-body cone is the feasibility of the block-diagonal LMI
\begin{equation}\label{eq:canonicalLMI}
F(\vec{\alpha},g)=
\begin{pmatrix} p_0 & \tfrac{1}{2}(p_1-g)\\[2pt] \tfrac{1}{2}(p_1-g) & p_2+g\end{pmatrix}
\oplus\,[\,g\,]\,\oplus\,[-\mathcal{h}_2(0,\vec{\alpha})]\,\oplus\,[-\mathcal{h}_2(1,\vec{\alpha})]\,\oplus\,[-\mathcal{h}_3(\vec{\alpha})]\;\succeq\;0,
\end{equation}
a single $6\times6$ linear matrix inequality, linear and homogeneous in the ten variables $(\vec{\alpha},g)$. The Gram parametrization of Eq.~\eqref{eq:sosrep} has four entries subject to the three matching equations of App.~\ref{app:sdpgram}, leaving exactly one free lifted variable: the weight $g$.

The lifted variable can be eliminated in closed form, at the price of leaving the LMI format. The determinant of the $2\times2$ block of Eq.~\eqref{eq:canonicalLMI} is a concave quadratic in $g$ whose discriminant collapses to $4\,p(0)p(1)$, so the block is positive semidefinite exactly for $g\in[g_-,g_+]$ with
\begin{equation}\label{eq:gfiber}
g_{\pm}=(2p_0+p_1)\pm2\sqrt{p(0)\,p(1)},
\end{equation}
and $\vec{\alpha}$ belongs to the three-body cone if and only if
\begin{equation}\label{eq:solverfree}
\mathcal{h}_2(0,\vec{\alpha})\le0,\quad \mathcal{h}_2(1,\vec{\alpha})\le0,\quad \mathcal{h}_3(\vec{\alpha})\le0,\quad
p(0)\ge0,\quad p(1)\ge0,\quad 2p_0+p_1+2\sqrt{p(0)\,p(1)}\ge0,
\end{equation}
a solver-free membership test.
The semidefinite form matters for optimization (Sec.~\ref{sec:sdp:opt}) and for the structural statement below, and the square root in Eq.~\eqref{eq:solverfree} is the trace left behind by the projection. The boundary strata of the cone are the degenerations of the fiber $[\max(0,g_-),\,g_+]$ to a single point.
On the endpoint facets $p(0)=0$ or $p(1)=0$ the certificate is unique, $g=2p_0+p_1$, while on the interior conic of Theorem~\ref{thm:QuadraticSurface}, where $p$ has a double root in $(0,1)$, the fiber is pinned at $g=0$: a pure square, with no interval weight.

\subsection{Necessity of the lift}\label{app:lift:necessity}

Let
\begin{equation}
P_2:=\{(p_0,p_1,p_2)\in\mathbbm{R}^{3} : p_0+p_1x+p_2x^{2}\ge0\ \text{ for all }x\in[0,1]\},
\end{equation}
the closed, pointed, full-dimensional cone of quadratics nonnegative on the interval. The first off-diagonal condition is its preimage under $\vec{\alpha}\mapsto(p_0,p_1,p_2)$. We freely identify a quadratic with its coefficient vector, e.g.\ $(1-x)^{2}\equiv(1,-2,1)$ and $x(1-x)\equiv(0,1,-1)$ for $0\neq p\in P_2$.

\begin{theorem}\label{thm:notspectrahedron}
$P_2$ is a projected spectrahedron (Lemma~\ref{thm:sos}), but it is not itself a spectrahedron.
\end{theorem}
\begin{proof}
We prove the above by contradiction: i.e., we assume that $P_{2}$ is a spectrahedron and we will show that its implications contradict a theorem by Ramana and Goldman~\cite{ramana_geometric_1995} which states that every extreme point of a spectrahedron is exposed.

First, let us consider the following set
\begin{equation}
B=P_{2}\cap \left\{(p_{0},p_{1},p_{2})\in \mathbb{R}^{3}: p_{0}=1\right\}.
\end{equation}
By construction, $B$ is an affine slice of $P_{2}$, and so if $P_{2}$ is spectrahedron, so is $B$. 

Next, let us consider a point $q\coloneqq  (1,-2,1)\in B$ which is associated with the polynomial $q(x)=(1-x)^{2}$. We will now show by contradiction that this point is extreme in $B$. Let $q(x)=q_{1}(x)+q_{2}(x)$ for some $q_{1}(x),q_{2}(x)\in B$. Notice that $q(1)=0$ implies that each $q_{1}(1),q_{2}(1)=0$ since $q_{1}(x),q_{2}(x)\geqslant 0$ for $x\in [0,1]$. Next, observe that for any function $f(x)\in B$ satisfying $f(1)=0$ also obeys $f'(1)\leqslant 0$, since $f(x)$ is nonnegative on $[0,1]$ and $f'(1)\geqslant 0$ would imply that $f(1-\epsilon)<0$ for some small $\epsilon>0$. Then, $q'(1)=0$ gives $q_{1}'(1)+q_{2}'(1)=0$, which by the virtue of the above argument, implies $q_{1}'(1),q_{2}'(1)=0$. The only quadratic with a double root at $1$ belonging to $B$ is $(1-x)^{2}$, and so $q_{1}(x),q_{2}(x)=q(x)$, i.e., $q(x)$ is an extreme point in $B$.

This then by Ramana-Goldman Theorem~\cite{ramana_geometric_1995} implies that $q(x)$ is exposed in $B$: there exists a linear function $h : \mathbb{R}^{3}\rightarrow \mathbb{R}$, such that $h(b)\geqslant0$ for all $b\in B$, with $h(b)=0$ implying $b=q(x)$. Such a function can always be written as
\begin{equation}
\chi(p_{1},p_{2})=a_{0}+a_{1}p_{1}+a_{2}p_{2},
\end{equation}
for some $a_{0},a_{1},a_{2}\in \mathbb{R}$. To determine the coefficients of the function $\chi$ we consider the family of points $q_{t}=(1,-2t,t^{2})$, with $q_{1}=q$. Notice, that for all $t\in \mathbb{R}$, $q_{t}\in B$ since it can be associated with a square polynomial $q_{t}(x)=(1-tx)^{2}$. Then, let $f(t)$ be a function defined as
\begin{equation}
f(t)\coloneqq \chi (-2t,t^{2})=a_{0}-2 t a_{1}+t^{2} a_{2}.
\end{equation}
Since we demand that $\chi(\cdot)$ achieves it's minimum value $0$ on the set $B$ on when evaluated on $q$, we can infer that the function $f(t)$ achieves its minimum at $t=1$
\begin{equation}
f'(1)=-2 a_{1}+ 2a_{2}=0,
\end{equation}
and so $a_{1}=a_{2}$. Then, since $\chi(q)=0$ we also get
\begin{equation}
f(1)=a_{0}-2 a_{1}+ a_{2} = 0,\quad \Rightarrow\quad a_{0}=a_{1}.
\end{equation}
Therefore, the function $\chi$ has a form
\begin{equation}
\chi(p_{1},p_{2})= a_{0}(1+p_{1}+p_{2}).
\end{equation}
Clearly, $s=(1,-3/2,1/2)\in B$ and $\chi(s)=0$, and so $q$ is not a unique root of $\chi$. Therefore, $q$ is extreme but not exposed in $B$, hence by the Ramana-Goldman Theorem~\cite{ramana_geometric_1995} $B$ is not a spectrahedron, hence, $P_{2}$ is also not a spectrahedron.
\end{proof}

\section{Endpoint-bounded subregion via Sturm's theorem}\label{app:boundary}

The semidefinite characterization of Sec.~\ref{sec:sdp} decides stoquasticity for every $\vec{\alpha}$ and every order $K$. For the subregion $\partial'\mathbbm{S}$ on which the supremum of $\mathcal{h}_{d}(x,\vec{\alpha})$ over $[0,1]$ is attained at an endpoint, so that the boundary is purely linear, Sturm's theorem gives an explicit local description of that boundary: around any generic $\vec{\alpha}$, one at which the polynomial signs entering the test do not vanish, the membership verdict is constant on an open ball, so the boundary is the locus where such a sign crosses zero. In particular, this certifies membership there without a semidefinite solver. The supremum lies at an endpoint precisely when the system
\begin{equation}
\label{eq:HplaneIneqSys}
\begin{aligned}
f(x,\vec{\alpha}) &\coloneqq \partial_x \mathcal{h}_{d}(x,\vec{\alpha}) = 0,\\
g_1(x,\vec{\alpha}) &\coloneqq \partial_x^2 \mathcal{h}_{d}(x,\vec{\alpha}) < 0,\\
g_2(x,\vec{\alpha}) &\coloneqq \mathcal{h}_{d}(x,\vec{\alpha}) - \mathcal{h}_{d}(0,\vec{\alpha}) > 0,\\
g_3(x,\vec{\alpha}) &\coloneqq \mathcal{h}_{d}(x,\vec{\alpha}) - \mathcal{h}_{d}(1,\vec{\alpha}) > 0
\end{aligned}
\end{equation}
has no solution in $(0,1)$. Since $\mathcal{h}_{d}(x,\vec{\alpha})$ is a polynomial in $x$ with coefficients linear in $\vec{\alpha}$ (Theorem~\ref{thm:BellMatElPol}), its interior critical points are the real roots of $f=\partial_x \mathcal{h}_{d}$, which Sturm's theorem counts exactly on any interval. Isolating each of the at most $K-d-1$ roots by bisection and testing the remaining inequalities, using $\gcd(f,g_j)$ to detect a common root of $f$ and $g_j$, decides infeasibility. Because the coefficients of $\mathcal{h}_{d}$ vary continuously with $\vec{\alpha}$, the verdict is constant until a critical point reaches an endpoint or some $g_j$ vanishes at a root.
The subregion $\partial'\mathbbm{S}$ is therefore traced by following $\vec{\alpha}$ between such events.

\end{document}